\documentclass[runningheads]{llncs}

\usepackage{xspace,cite,amsmath,comment,amssymb,subcaption,enumerate}

\usepackage[paperheight=235mm,paperwidth=155mm,textwidth=12.2cm,textheight=19.3cm]{geometry}
\usepackage{bbding}

\makeatletter
\RequirePackage[bookmarks,unicode,colorlinks=true]{hyperref}%
   \def\@citecolor{blue}%
   \def\@urlcolor{blue}%
   \def\@linkcolor{blue}%

\def\orcidID#1{\smash{\href{http://orcid.org/#1}{\protect\raisebox{-1.25pt}{\protect\includegraphics{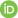}}}}}
\makeatother

\usepackage{cleveref}

\newcommand{\cat}[1]{\mathbf{#1}}

\newcommand{\id}[1][]{\ensuremath{\mathrm{id}_{#1}}}
\newcommand{\tuple}[1]{\mathopen{\langle}#1\mathclose{\rangle}}
\newcommand{\ie}{\text{i.e.,}\xspace}
\newcommand{\eg}{\text{e.g.,}\xspace}

\DeclareMathOperator{\dom}{dom}

\newcommand{\Set}{\cat{Set}}
\newcommand{\A}{\mathcal{A}}
\newcommand{\CC}{\cat{C}}

\newcommand{\res}{\int}

\newcommand{\CF}{\CC_F}
\newcommand{\vect}[1]{\bar{#1}}
\newcommand{\f}{\vect{f}}

\usepackage{tikz}
\usetikzlibrary{decorations.pathreplacing}
\usetikzlibrary{decorations.markings}
\usetikzlibrary{calc}
\usetikzlibrary{arrows}
\usetikzlibrary{arrows.meta}
\usetikzlibrary{shapes.geometric}
\tikzset{>={To[length=2.5pt,width=4pt]}}

\usetikzlibrary{intersections}

\newenvironment{pic}[1][]
{\begin{aligned}\begin{tikzpicture}[font=\tiny,#1]}
{\end{tikzpicture}\end{aligned}}

\pgfdeclarelayer{foreground}
\pgfdeclarelayer{background}
\pgfdeclarelayer{morphismlayer}
\pgfdeclarelayer{dashedmorphismlayer}
\pgfdeclarelayer{edgelayer}
\pgfdeclarelayer{nodelayer}
\pgfsetlayers{background,main,morphismlayer,dashedmorphismlayer,foreground,edgelayer,nodelayer}

\def\thickness{0.7pt}
\makeatletter
\pgfkeys{%
  /tikz/on layer/.code={
    \pgfonlayer{#1}\begingroup
    \aftergroup\endpgfonlayer
    \aftergroup\endgroup
  },
  /tikz/node on layer/.code={
    \gdef\node@@on@layer{%
      \setbox\tikz@tempbox=\hbox\bgroup\pgfonlayer{#1}\unhbox\tikz@tempbox\endpgfonlayer\pgfsetlinewidth{\thickness}\egroup}
    \aftergroup\node@on@layer
  },
  /tikz/end node on layer/.code={
    \endpgfonlayer\endgroup\endgroup
  }
}
\def\node@on@layer{\aftergroup\node@@on@layer}
\makeatother

\tikzstyle{braid}=[double=black, line width=3*\thickness, double distance=\thickness, draw=white, text=black]
\tikzset{every picture/.style={line width=\thickness, draw=black}}
\tikzstyle{pure}=[line width=.7pt]
\tikzstyle{string}=[line width=\thickness]
\tikzstyle{scalar}=[circle, inner sep=0pt, minimum width=15pt, draw, line width=\thickness, fill=white]
\tikzstyle{dot}=[circle, draw=black, fill=black!25, inner sep=.4ex, line width=\thickness, node on layer=foreground]
\tikzstyle{blackdot}=[circle, draw=black, fill=black!75, inner sep=.4ex, line width=\thickness, node on layer=foreground, text=white]
\tikzstyle{whitedot}=[circle, draw=black, fill=white, inner sep=.4ex, line width=\thickness, node on layer=foreground]
\tikzstyle{mixedmorphism}=[morphism, minimum width=30pt, minimum height=16pt, draw, font=\small, inner sep=0pt, fill=white, line width=\thickness,rounded corners=1ex]

\tikzstyle{triangle} = [regular polygon, regular polygon sides=3, draw=black, fill=black!20,scale=0.4, node on layer=foreground]

\tikzstyle{whitetriangle}=[triangle, fill=white]
\tikzstyle{greytriangle}=[triangle, fill=black!20]
\tikzstyle{darkgreytriangle}=[triangle, fill=black!50]
\tikzstyle{blacktriangle}=[triangle, fill=black]

\tikzstyle{invertedtriangle} = [triangle,scale=-1]
\tikzstyle{whiteinvertedtriangle}=[invertedtriangle, fill=white]
\tikzstyle{greyinvertedtriangle}=[invertedtriangle, fill=black!20]
\tikzstyle{darkgreyinvertedtriangle}=[invertedtriangle, fill=black!50]
\tikzstyle{blackinvertedtriangle}=[invertedtriangle, fill=black]

\tikzset{functor1/.style={gray!50}}
\tikzset{functor2/.style={gray!50}}
\tikzstyle{thick}=[line width=\thickness]
\tikzstyle{tiny}=[font=\tiny]
\tikzset{circlelabel/.style={draw, thick, circle, inner sep=-5pt,
 fill=white, minimum width=14pt, fill opacity=1, node on layer=foreground, font=\scriptsize}}
\tikzset{shade 1 transparent/.style={fill=black!20, fill opacity=0.8}}
\tikzset{shade 2 transparent/.style={fill=black!35, fill opacity=0.8}}
\tikzset{shade 3 transparent/.style={fill=black!50, fill opacity=0.8}}
\tikzset{shade 4 transparent/.style={fill=black!65, fill opacity=0.8}}
\tikzset{shade 1/.style={fill=black!16, fill opacity=1}}
\tikzset{shade 2/.style={fill=black!28, fill opacity=1}}
\tikzset{shade 3/.style={fill=black!40, fill opacity=1}}
\tikzset{shade 4/.style={fill=black!52, fill opacity=1}}

\def\strarr{line width=0.7pt, length=4pt, width=5pt, color=black}

\tikzset{arrow/.style={decoration={
    markings,
    mark=at position #1 with \arrow{>[\strarr]}},
    postaction=decorate},
    reverse arrow/.style={decoration={
    markings,
    mark=at position #1 with {{\arrow{<[\strarr]}}}},
    postaction=decorate}
}

\tikzset{overbrace/.style={
     decoration={brace},
     decorate}
}
\tikzset{underbrace/.style={
     decoration={brace, mirror},
     decorate}
}

\newif\ifblack\pgfkeys{/tikz/black/.is if=black}
\newif\ifwedge\pgfkeys{/tikz/wedge/.is if=wedge}
\newif\ifvflip\pgfkeys{/tikz/vflip/.is if=vflip}
\newif\ifhflip\pgfkeys{/tikz/hflip/.is if=hflip}
\newif\ifhvflip\pgfkeys{/tikz/hvflip/.is if=hvflip}
\newif\ifconnectsw\pgfkeys{/tikz/connect sw/.is if=connectsw}
\newif\ifconnectse\pgfkeys{/tikz/connect se/.is if=connectse}
\newif\ifconnectn\pgfkeys{/tikz/connect n/.is if=connectn}
\newif\ifconnects\pgfkeys{/tikz/connect s/.is if=connects}
\newif\ifconnectnw\pgfkeys{/tikz/connect nw/.is if=connectnw}
\newif\ifconnectne\pgfkeys{/tikz/connect ne/.is if=connectne}
\newif\ifconnectnwf\pgfkeys{/tikz/connect nw >/.is if=connectnwf}
\newif\ifconnectnef\pgfkeys{/tikz/connect ne >/.is if=connectnef}
\newif\ifconnectswf\pgfkeys{/tikz/connect sw >/.is if=connectswf}
\newif\ifconnectsef\pgfkeys{/tikz/connect se >/.is if=connectsef}
\newif\ifconnectnf\pgfkeys{/tikz/connect n >/.is if=connectnf}
\newif\ifconnectsf\pgfkeys{/tikz/connect s >/.is if=connectsf}
\newif\ifconnectnwr\pgfkeys{/tikz/connect nw </.is if=connectnwr}
\newif\ifconnectner\pgfkeys{/tikz/connect ne </.is if=connectner}
\newif\ifconnectswr\pgfkeys{/tikz/connect sw </.is if=connectswr}
\newif\ifconnectser\pgfkeys{/tikz/connect se </.is if=connectser}
\newif\ifconnectnr\pgfkeys{/tikz/connect n </.is if=connectnr}
\newif\ifconnectsr\pgfkeys{/tikz/connect s </.is if=connectsr}
\tikzset{keylengthnw/.initial=\connectheight}
\tikzset{keylengthn/.initial =\connectheight}
\tikzset{keylengthne/.initial=\connectheight}
\tikzset{keylengthsw/.initial=\connectheight}
\tikzset{keylengths/.initial =\connectheight}
\tikzset{keylengthse/.initial=\connectheight}
\tikzset{connect nw length/.style={connect nw=true, keylengthnw={#1}}}
\tikzset{connect n length/.style ={connect n =true, keylengthn ={#1}}}
\tikzset{connect ne length/.style={connect ne=true, keylengthne={#1}}}
\tikzset{connect sw length/.style={connect sw=true, keylengthsw={#1}}}
\tikzset{connect s length/.style ={connect s =true, keylengths ={#1}}}
\tikzset{connect se length/.style={connect se=true, keylengthse={#1}}}
\tikzset{connect nw < length/.style={connect nw <=true, keylengthnw={#1}}}
\tikzset{connect n < length/.style ={connect n <=true,  keylengthn ={#1}}}
\tikzset{connect ne < length/.style={connect ne <=true, keylengthne={#1}}}
\tikzset{connect sw < length/.style={connect sw <=true, keylengthnw={#1}}}
\tikzset{connect s < length/.style ={connect s <=true,  keylengths ={#1}}}
\tikzset{connect se < length/.style={connect se <=true, keylengthse={#1}}}
\tikzset{connect nw > length/.style={connect nw >=true, keylengthnw={#1}}}
\tikzset{connect n > length/.style ={connect n >=true,  keylengthn ={#1}}}
\tikzset{connect ne > length/.style={connect ne >=true, keylengthne={#1}}}
\tikzset{connect sw > length/.style={connect sw >=true, keylengthsw={#1}}}
\tikzset{connect s > length/.style ={connect s >=true,  keylengths ={#1}}}
\tikzset{connect se > length/.style={connect se >=true, keylengthse={#1}}}

\newlength\morphismheight
\newlength\wedgewidth
\newlength\minimummorphismwidth
\newlength\stateheight
\newlength\minimumstatewidth
\newlength\connectheight
\tikzset{width/.initial=\minimummorphismwidth}
\tikzset{colour/.initial=white}

\makeatletter
\pgfarrowsdeclare{thickarrow}{thickarrow}
{
  \pgfutil@tempdima=-0.84pt%
  \advance\pgfutil@tempdima by-1.3\pgflinewidth%
  \pgfutil@tempdimb=-1.7pt%
  \advance\pgfutil@tempdimb by.625\pgflinewidth%
  \pgfarrowsleftextend{+\pgfutil@tempdima}
  \pgfarrowsrightextend{+\pgfutil@tempdimb}
}
{
  \pgfmathparse{\pgfgetarrowoptions{thickarrow}}%
  \pgfsetlinewidth{1.25 pt}
  \pgfutil@tempdima=0.28pt%
  \advance\pgfutil@tempdima by.3\pgflinewidth%
  \pgfsetlinewidth{0.8\pgflinewidth}
  \pgfsetdash{}{+0pt}
  \pgfsetroundcap
  \pgfsetroundjoin
  \pgfpathmoveto{\pgfqpoint{-3\pgfutil@tempdima}{4\pgfutil@tempdima}}
  \pgfpathcurveto
  {\pgfqpoint{-2.75\pgfutil@tempdima}{2.5\pgfutil@tempdima}}
  {\pgfqpoint{0pt}{0.25\pgfutil@tempdima}}
  {\pgfqpoint{0.75\pgfutil@tempdima}{0pt}}
  \pgfpathcurveto
  {\pgfqpoint{0pt}{-0.25\pgfutil@tempdima}}
  {\pgfqpoint{-2.75\pgfutil@tempdima}{-2.5\pgfutil@tempdima}}
  {\pgfqpoint{-3\pgfutil@tempdima}{-4\pgfutil@tempdima}}
  \pgfusepathqstroke
}
\pgfarrowsdeclare{reversethickarrow}{reversethickarrow}
{
  \pgfutil@tempdima=-0.84pt%
  \advance\pgfutil@tempdima by-1.3\pgflinewidth%
  \pgfutil@tempdimb=0.2pt%
  \advance\pgfutil@tempdimb by.625\pgflinewidth%
  \pgfarrowsleftextend{+\pgfutil@tempdima}
  \pgfarrowsrightextend{+\pgfutil@tempdimb}
}
{
  \pgftransformxscale{-1}
  \pgfmathparse{\pgfgetarrowoptions{thickarrow}}%
  \ifpgfmathunitsdeclared%
    \pgfmathparse{\pgfmathresult pt}%
  \else%
    \pgfmathparse{\pgfmathresult*\pgflinewidth}%
  \fi%
  \let\thickness=\pgfmathresult
  \pgfsetlinewidth{1.25 pt}
  \pgfutil@tempdima=0.28pt%
  \advance\pgfutil@tempdima by.3\pgflinewidth%
  \pgfsetlinewidth{0.8\pgflinewidth}
  \pgfsetdash{}{+0pt}
  \pgfsetroundcap
  \pgfsetroundjoin
  \pgfpathmoveto{\pgfqpoint{-3\pgfutil@tempdima}{4\pgfutil@tempdima}}
  \pgfpathcurveto
  {\pgfqpoint{-2.75\pgfutil@tempdima}{2.5\pgfutil@tempdima}}
  {\pgfqpoint{0pt}{0.25\pgfutil@tempdima}}
  {\pgfqpoint{0.75\pgfutil@tempdima}{0pt}}
  \pgfpathcurveto
  {\pgfqpoint{0pt}{-0.25\pgfutil@tempdima}}
  {\pgfqpoint{-2.75\pgfutil@tempdima}{-2.5\pgfutil@tempdima}}
  {\pgfqpoint{-3\pgfutil@tempdima}{-4\pgfutil@tempdima}}
  \pgfusepathqstroke
}
\makeatother

\tikzset{diredge/.style={decoration={
  markings,
  mark=at position 0.525 with {\arrow{#1}}},postaction={decorate}}}
\tikzset{
    diredge/.default=>
}

\tikzset{diredgestart/.style={decoration={
  markings,
  mark=at position 4pt with {\arrow{#1}}},postaction={decorate}}}
\tikzset{
    diredgestart/.default=<
}

\tikzset{diredgeend/.style={decoration={
  markings,
  mark=at position 1 with {\arrow{#1}}},postaction={decorate}}}
\tikzset{
    diredgeend/.default=>
}

\makeatletter
\pgfdeclareshape{ground}
{
    \savedanchor\centerpoint
    {
        \pgf@x=0pt
        \pgf@y=0pt
    }
    \anchor{center}{\centerpoint}
    \anchorborder{\centerpoint}
    \anchor{north}
    {
        \pgf@x=0pt
        \pgf@y=0.5*0.33*\stateheight
    }
    \anchor{south}
    {
        \pgf@x=0pt
        \pgf@y=0pt
    }
    \saveddimen\overallwidth
    {
        \pgfkeysgetvalue{/pgf/minimum width}{\minwidth}
        \pgf@x=\minimumstatewidth
        \ifdim\pgf@x<\minwidth
            \pgf@x=\minwidth
        \fi
    }
    \backgroundpath
    {
        \begin{pgfonlayer}{foreground}
        \pgfsetstrokecolor{black}
        \pgfsetlinewidth{1.25pt}
        \ifhflip
            \pgftransformyscale{-1}
        \fi
        \pgftransformscale{0.5}
        \pgfpathmoveto{\pgfpoint{-0.5*\overallwidth}{0}}
        \pgfpathlineto{\pgfpoint{0.5*\overallwidth}{0}}
        \pgfpathmoveto{\pgfpoint{-0.33*\overallwidth}{0.33*\stateheight}}
        \pgfpathlineto{\pgfpoint{0.33*\overallwidth}{0.33*\stateheight}}
        \pgfpathmoveto{\pgfpoint{-0.16*\overallwidth}{0.66*\stateheight}}
        \pgfpathlineto{\pgfpoint{0.16*\overallwidth}{0.66*\stateheight}}
        \pgfpathmoveto{\pgfpoint{-0.02*\overallwidth}{\stateheight}}
        \pgfpathlineto{\pgfpoint{0.02*\overallwidth}{\stateheight}}
        \pgfusepath{stroke}
        \end{pgfonlayer}
    }
}
\tikzset{forward arrow style/.style={every to/.style, decoration={
    markings,
    mark=at position 0.5*\pgfdecoratedpathlength+2pt with \arrow{>[\strarr]}},
    postaction=decorate}}
\tikzset{reverse arrow style/.style={every to/.style, decoration={
    markings,
    mark=at position 0.5*\pgfdecoratedpathlength+2pt with \arrow{<[\strarr]}},
    postaction=decorate}}
\pgfdeclareshape{morphismshape}
{
    \savedanchor\centerpoint
    {
        \pgf@x=0pt
        \pgf@y=0pt
    }
    \anchor{center}{\centerpoint}
    \anchorborder{\centerpoint}
    \saveddimen\savedlengthnw
    {
        \pgfkeysgetvalue{/tikz/keylengthnw}{\len}
        \pgf@x=\len
    }
    \saveddimen\savedlengthn
    {
        \pgfkeysgetvalue{/tikz/keylengthn}{\len}
        \pgf@x=\len
    }
    \saveddimen\savedlengthne
    {
        \pgfkeysgetvalue{/tikz/keylengthne}{\len}
        \pgf@x=\len
    }
    \saveddimen\savedlengthsw
    {
        \pgfkeysgetvalue{/tikz/keylengthsw}{\len}
        \pgf@x=\len
    }
    \saveddimen\savedlengths
    {
        \pgfkeysgetvalue{/tikz/keylengths}{\len}
        \pgf@x=\len
    }
    \saveddimen\savedlengthse
    {
        \pgfkeysgetvalue{/tikz/keylengthse}{\len}
        \pgf@x=\len
    }
    \saveddimen\overallwidth
    {
        \pgfkeysgetvalue{/tikz/width}{\minwidth}
        \pgf@x=\wd\pgfnodeparttextbox
        \ifdim\pgf@x<\minwidth
            \pgf@x=\minwidth
        \fi
    }
    \savedanchor{\upperrightcorner}
    {
        \pgf@y=.5\ht\pgfnodeparttextbox
        \advance\pgf@y by -.5\dp\pgfnodeparttextbox
        \pgf@x=.5\wd\pgfnodeparttextbox
    }
    \anchor{north}
    {
        \pgf@x=0pt
        \pgf@y=0.5\morphismheight
    }
    \anchor{north east}
    {
        \pgf@x=\overallwidth
        \multiply \pgf@x by 2
        \divide \pgf@x by 5
        \pgf@y=0.5\morphismheight
    }
    \anchor{east}
    {
        \pgf@x=\overallwidth
        \divide \pgf@x by 2
        \advance \pgf@x by 5pt
        \pgf@y=0pt
    }
    \anchor{west}
    {
        \pgf@x=-\overallwidth
        \divide \pgf@x by 2
        \advance \pgf@x by -5pt
        \pgf@y=0pt
    }
    \anchor{north west}
    {
        \pgf@x=-\overallwidth
        \multiply \pgf@x by 2
        \divide \pgf@x by 5
        \pgf@y=0.5\morphismheight
    }
    \anchor{connect nw}
    {
        \pgf@x=-\overallwidth
        \multiply \pgf@x by 2
        \divide \pgf@x by 5
        \pgf@y=0.5\morphismheight
        \advance\pgf@y by \savedlengthnw
    }
    \anchor{connect ne}
    {
        \pgf@x=\overallwidth
        \multiply \pgf@x by 2
        \divide \pgf@x by 5
        \pgf@y=0.5\morphismheight
        \advance\pgf@y by \savedlengthne
    }
    \anchor{connect sw}
    {
        \pgf@x=-\overallwidth
        \multiply \pgf@x by 2
        \divide \pgf@x by 5
        \pgf@y=-0.5\morphismheight
        \advance\pgf@y by -\savedlengthsw
    }
    \anchor{connect se}
    {
        \pgf@x=\overallwidth
        \multiply \pgf@x by 2
        \divide \pgf@x by 5
        \pgf@y=-0.5\morphismheight
        \advance\pgf@y by -\savedlengthse
    }
    \anchor{connect n}
    {
        \pgf@x=0pt
        \pgf@y=0.5\morphismheight
        \advance\pgf@y by \savedlengthn
    }
    \anchor{connect s}
    {
        \pgf@x=0pt
        \pgf@y=-0.5\morphismheight
        \advance\pgf@y by -\savedlengths
    }
    \anchor{south east}
    {
        \pgf@x=\overallwidth
        \multiply \pgf@x by 2
        \divide \pgf@x by 5
        \pgf@y=-0.5\morphismheight
    }
    \anchor{south west}
    {
        \pgf@x=-\overallwidth
        \multiply \pgf@x by 2
        \divide \pgf@x by 5
        \pgf@y=-0.5\morphismheight
    }
    \anchor{south}
    {
        \pgf@x=0pt
        \pgf@y=-0.5\morphismheight
    }
    \anchor{text}
    {
        \upperrightcorner
        \pgf@x=-\pgf@x
        \pgf@y=-\pgf@y
    }
    \backgroundpath
    {
    \begin{scope}
        \pgfkeysgetvalue{/tikz/fill}{\morphismfill}
        \pgfsetstrokecolor{black}
        \pgfsetlinewidth{\thickness}
        \begin{scope}
            \begin{pgfonlayer}{morphismlayer}
        \pgfsetstrokecolor{black}
        \pgfsetfillcolor{\pgfkeysvalueof{/tikz/colour}}
        \pgfsetlinewidth{\thickness}
                \ifhflip
                    \pgftransformyscale{-1}
                \fi
                \ifvflip
                    \pgftransformxscale{-1}
                \fi
                \ifhvflip
                    \pgftransformxscale{-1}
                    \pgftransformyscale{-1}
                \fi
                \pgfpathmoveto{\pgfpoint
                    {-0.5*\overallwidth-5pt}
                    {0.5*\morphismheight}}
                \pgfpathlineto{\pgfpoint
                    {0.5*\overallwidth+5pt}
                    {0.5*\morphismheight}}
                \ifwedge
                    \pgfpathlineto{\pgfpoint
                        {0.5*\overallwidth + \wedgewidth}
                        {-0.5*\morphismheight}}
                \else
                    \pgfpathlineto{\pgfpoint
                        {0.5*\overallwidth + 5pt}
                        {-0.5*\morphismheight}}
                \fi
                \pgfpathlineto{\pgfpoint
                    {-0.5*\overallwidth-5pt}
                    {-0.5*\morphismheight}}
                \pgfpathclose
                \pgfusepath{fill,stroke}
            \end{pgfonlayer}
        \end{scope}
        \ifconnectnw
            \pgfpathmoveto{\pgfpoint
                {-0.4*\overallwidth}
                {0.5*\morphismheight}}
            \pgfpathlineto{\pgfpoint
                {-0.4*\overallwidth}
                {0.5*\morphismheight+\savedlengthnw}}
            \pgfusepath{stroke}
        \fi
        \ifconnectne
            \pgfpathmoveto{\pgfpoint
                {0.4*\overallwidth}
                {0.5*\morphismheight}}
            \pgfpathlineto{\pgfpoint
                {0.4*\overallwidth}
                {0.5*\morphismheight+\savedlengthne}}
            \pgfusepath{stroke}
        \fi
        \ifconnectsw
            \pgfpathmoveto{\pgfpoint
                {-0.4*\overallwidth}
                {-0.5*\morphismheight}}
            \pgfpathlineto{\pgfpoint
                {-0.4*\overallwidth}
                {-0.5*\morphismheight-\savedlengthsw}}
            \pgfusepath{stroke}
        \fi
        \ifconnectse
            \pgfpathmoveto{\pgfpoint
                {0.4*\overallwidth}
                {-0.5*\morphismheight}}
            \pgfpathlineto{\pgfpoint
                {0.4*\overallwidth}
                {-0.5*\morphismheight-\savedlengthse}}
            \pgfusepath{stroke}
        \fi
        \ifconnectn
            \pgfpathmoveto{\pgfpoint
                {0pt}
                {0.5*\morphismheight}}
            \pgfpathlineto{\pgfpoint
                {0pt}
                {0.5*\morphismheight+\savedlengthn}}
            \pgfusepath{stroke}
        \fi
        \ifconnects
            \pgfpathmoveto{\pgfpoint
                {0pt}
                {-0.5*\morphismheight}}
            \pgfpathlineto{\pgfpoint
                {0pt}
                {-0.5*\morphismheight-\savedlengths}}
            \pgfusepath{stroke}
        \fi
        \ifconnectnwf
            \draw [forward arrow style] (-0.4*\overallwidth,0.5*\morphismheight)
                to (-0.4*\overallwidth,0.5*\morphismheight+\savedlengthnw);
        \fi
        \ifconnectnef
            \draw [forward arrow style] (0.4*\overallwidth,0.5*\morphismheight)
                to (0.4*\overallwidth,0.5*\morphismheight+\savedlengthne);
        \fi
        \ifconnectswf
            \draw [forward arrow style] (-0.4*\overallwidth,-0.5*\morphismheight-\savedlengthsw)
                to (-0.4*\overallwidth,-0.5*\morphismheight);
        \fi
        \ifconnectsef
            \draw [forward arrow style] (0.4*\overallwidth,-0.5*\morphismheight-\savedlengthse)
                to (0.4*\overallwidth,-0.5*\morphismheight);
        \fi
        \ifconnectnf
            \draw [forward arrow style] (0,0.5*\morphismheight)
                to (0,0.5*\morphismheight+\savedlengthn);
        \fi
        \ifconnectsf
            \draw [forward arrow style] (0,-0.5*\morphismheight-\savedlengths)
                to (0,-0.5*\morphismheight);
        \fi
        \ifconnectnwr
            \draw [reverse arrow style] (-0.4*\overallwidth,0.5*\morphismheight)
                to (-0.4*\overallwidth,0.5*\morphismheight+\savedlengthnw);
        \fi
        \ifconnectner
            \draw [reverse arrow style] (0.4*\overallwidth,0.5*\morphismheight)
                to (0.4*\overallwidth,0.5*\morphismheight+\savedlengthne);
        \fi
        \ifconnectswr
            \draw [reverse arrow style] (-0.4*\overallwidth,-0.5*\morphismheight-\savedlengthsw)
                to (-0.4*\overallwidth,-0.5*\morphismheight);
        \fi
        \ifconnectser
            \draw [reverse arrow style] (0.4*\overallwidth,-0.5*\morphismheight-\savedlengthse)
                to (0.4*\overallwidth,-0.5*\morphismheight);
        \fi
        \ifconnectnr
            \draw [reverse arrow style] (0,0.5*\morphismheight)
                to (0,0.5*\morphismheight+\savedlengthn);
        \fi
        \ifconnectsr
            \draw [reverse arrow style] (0,-0.5*\morphismheight-\savedlengths)
                to (0,-0.5*\morphismheight);
        \fi
    \end{scope}
    }
}
\tikzset{morphism/.style={morphismshape, node on layer=foreground}}
\pgfdeclareshape{dashedmorphismshape}
{
    \savedanchor\centerpoint
    {
        \pgf@x=0pt
        \pgf@y=0pt
    }
    \anchor{center}{\centerpoint}
    \anchorborder{\centerpoint}
    \saveddimen\savedlengthnw
    {
        \pgfkeysgetvalue{/tikz/keylengthnw}{\len}
        \pgf@x=\len
    }
    \saveddimen\savedlengthn
    {
        \pgfkeysgetvalue{/tikz/keylengthn}{\len}
        \pgf@x=\len
    }
    \saveddimen\savedlengthne
    {
        \pgfkeysgetvalue{/tikz/keylengthne}{\len}
        \pgf@x=\len
    }
    \saveddimen\savedlengthsw
    {
        \pgfkeysgetvalue{/tikz/keylengthsw}{\len}
        \pgf@x=\len
    }
    \saveddimen\savedlengths
    {
        \pgfkeysgetvalue{/tikz/keylengths}{\len}
        \pgf@x=\len
    }
    \saveddimen\savedlengthse
    {
        \pgfkeysgetvalue{/tikz/keylengthse}{\len}
        \pgf@x=\len
    }
    \saveddimen\overallwidth
    {
        \pgfkeysgetvalue{/tikz/width}{\minwidth}
        \pgf@x=\wd\pgfnodeparttextbox
        \ifdim\pgf@x<\minwidth
            \pgf@x=\minwidth
        \fi
    }
    \savedanchor{\upperrightcorner}
    {
        \pgf@y=.5\ht\pgfnodeparttextbox
        \advance\pgf@y by -.5\dp\pgfnodeparttextbox
        \pgf@x=.5\wd\pgfnodeparttextbox
    }
    \anchor{north}
    {
        \pgf@x=0pt
        \pgf@y=0.5\morphismheight
    }
    \anchor{north east}
    {
        \pgf@x=\overallwidth
        \multiply \pgf@x by 2
        \divide \pgf@x by 5
        \pgf@y=0.5\morphismheight
    }
    \anchor{east}
    {
        \pgf@x=\overallwidth
        \divide \pgf@x by 2
        \advance \pgf@x by 5pt
        \pgf@y=0pt
    }
    \anchor{west}
    {
        \pgf@x=-\overallwidth
        \divide \pgf@x by 2
        \advance \pgf@x by -5pt
        \pgf@y=0pt
    }
    \anchor{north west}
    {
        \pgf@x=-\overallwidth
        \multiply \pgf@x by 2
        \divide \pgf@x by 5
        \pgf@y=0.5\morphismheight
    }
    \anchor{connect nw}
    {
        \pgf@x=-\overallwidth
        \multiply \pgf@x by 2
        \divide \pgf@x by 5
        \pgf@y=0.5\morphismheight
        \advance\pgf@y by \savedlengthnw
    }
    \anchor{connect ne}
    {
        \pgf@x=\overallwidth
        \multiply \pgf@x by 2
        \divide \pgf@x by 5
        \pgf@y=0.5\morphismheight
        \advance\pgf@y by \savedlengthne
    }
    \anchor{connect sw}
    {
        \pgf@x=-\overallwidth
        \multiply \pgf@x by 2
        \divide \pgf@x by 5
        \pgf@y=-0.5\morphismheight
        \advance\pgf@y by -\savedlengthsw
    }
    \anchor{connect se}
    {
        \pgf@x=\overallwidth
        \multiply \pgf@x by 2
        \divide \pgf@x by 5
        \pgf@y=-0.5\morphismheight
        \advance\pgf@y by -\savedlengthse
    }
    \anchor{connect n}
    {
        \pgf@x=0pt
        \pgf@y=0.5\morphismheight
        \advance\pgf@y by \savedlengthn
    }
    \anchor{connect s}
    {
        \pgf@x=0pt
        \pgf@y=-0.5\morphismheight
        \advance\pgf@y by -\savedlengths
    }
    \anchor{south east}
    {
        \pgf@x=\overallwidth
        \multiply \pgf@x by 2
        \divide \pgf@x by 5
        \pgf@y=-0.5\morphismheight
    }
    \anchor{south west}
    {
        \pgf@x=-\overallwidth
        \multiply \pgf@x by 2
        \divide \pgf@x by 5
        \pgf@y=-0.5\morphismheight
    }
    \anchor{south}
    {
        \pgf@x=0pt
        \pgf@y=-0.5\morphismheight
    }
    \anchor{text}
    {
        \upperrightcorner
        \pgf@x=-\pgf@x
        \pgf@y=-\pgf@y
    }
    \backgroundpath
    {
    \begin{scope}
        \pgfkeysgetvalue{/tikz/fill}{\morphismfill}
        \pgfsetstrokecolor{black}
        \pgfsetlinewidth{\thickness}
        \begin{scope}
            \begin{pgfonlayer}{dashedmorphismlayer}
        \pgfsetstrokecolor{black}
        \pgfsetdash{{3pt}{3pt}}{0pt}
        \pgfsetfillcolor{\pgfkeysvalueof{/tikz/colour}}
        \pgfsetlinewidth{\thickness}
                \ifhflip
                    \pgftransformyscale{-1}
                \fi
                \ifvflip
                    \pgftransformxscale{-1}
                \fi
                \ifhvflip
                    \pgftransformxscale{-1}
                    \pgftransformyscale{-1}
                \fi
                \pgfpathmoveto{\pgfpoint
                    {-0.5*\overallwidth-5pt}
                    {0.5*\morphismheight}}
                \pgfpathlineto{\pgfpoint
                    {0.5*\overallwidth+5pt}
                    {0.5*\morphismheight}}
                \ifwedge
                    \pgfpathlineto{\pgfpoint
                        {0.5*\overallwidth + \wedgewidth}
                        {-0.5*\morphismheight}}
                \else
                    \pgfpathlineto{\pgfpoint
                        {0.5*\overallwidth + 5pt}
                        {-0.5*\morphismheight}}
                \fi
                \pgfpathlineto{\pgfpoint
                    {-0.5*\overallwidth-5pt}
                    {-0.5*\morphismheight}}
                \pgfpathclose
                \pgfusepath{fill,stroke}
            \end{pgfonlayer}
        \end{scope}
        \ifconnectnw
            \pgfpathmoveto{\pgfpoint
                {-0.4*\overallwidth}
                {0.5*\morphismheight}}
            \pgfpathlineto{\pgfpoint
                {-0.4*\overallwidth}
                {0.5*\morphismheight+\savedlengthnw}}
            \pgfusepath{stroke}
        \fi
        \ifconnectne
            \pgfpathmoveto{\pgfpoint
                {0.4*\overallwidth}
                {0.5*\morphismheight}}
            \pgfpathlineto{\pgfpoint
                {0.4*\overallwidth}
                {0.5*\morphismheight+\savedlengthne}}
            \pgfusepath{stroke}
        \fi
        \ifconnectsw
            \pgfpathmoveto{\pgfpoint
                {-0.4*\overallwidth}
                {-0.5*\morphismheight}}
            \pgfpathlineto{\pgfpoint
                {-0.4*\overallwidth}
                {-0.5*\morphismheight-\savedlengthsw}}
            \pgfusepath{stroke}
        \fi
        \ifconnectse
            \pgfpathmoveto{\pgfpoint
                {0.4*\overallwidth}
                {-0.5*\morphismheight}}
            \pgfpathlineto{\pgfpoint
                {0.4*\overallwidth}
                {-0.5*\morphismheight-\savedlengthse}}
            \pgfusepath{stroke}
        \fi
        \ifconnectn
            \pgfpathmoveto{\pgfpoint
                {0pt}
                {0.5*\morphismheight}}
            \pgfpathlineto{\pgfpoint
                {0pt}
                {0.5*\morphismheight+\savedlengthn}}
            \pgfusepath{stroke}
        \fi
        \ifconnects
            \pgfpathmoveto{\pgfpoint
                {0pt}
                {-0.5*\morphismheight}}
            \pgfpathlineto{\pgfpoint
                {0pt}
                {-0.5*\morphismheight-\savedlengths}}
            \pgfusepath{stroke}
        \fi
        \ifconnectnwf
            \draw [forward arrow style] (-0.4*\overallwidth,0.5*\morphismheight)
                to (-0.4*\overallwidth,0.5*\morphismheight+\savedlengthnw);
        \fi
        \ifconnectnef
            \draw [forward arrow style] (0.4*\overallwidth,0.5*\morphismheight)
                to (0.4*\overallwidth,0.5*\morphismheight+\savedlengthne);
        \fi
        \ifconnectswf
            \draw [forward arrow style] (-0.4*\overallwidth,-0.5*\morphismheight-\savedlengthsw)
                to (-0.4*\overallwidth,-0.5*\morphismheight);
        \fi
        \ifconnectsef
            \draw [forward arrow style] (0.4*\overallwidth,-0.5*\morphismheight-\savedlengthse)
                to (0.4*\overallwidth,-0.5*\morphismheight);
        \fi
        \ifconnectnf
            \draw [forward arrow style] (0,0.5*\morphismheight)
                to (0,0.5*\morphismheight+\savedlengthn);
        \fi
        \ifconnectsf
            \draw [forward arrow style] (0,-0.5*\morphismheight-\savedlengths)
                to (0,-0.5*\morphismheight);
        \fi
        \ifconnectnwr
            \draw [reverse arrow style] (-0.4*\overallwidth,0.5*\morphismheight)
                to (-0.4*\overallwidth,0.5*\morphismheight+\savedlengthnw);
        \fi
        \ifconnectner
            \draw [reverse arrow style] (0.4*\overallwidth,0.5*\morphismheight)
                to (0.4*\overallwidth,0.5*\morphismheight+\savedlengthne);
        \fi
        \ifconnectswr
            \draw [reverse arrow style] (-0.4*\overallwidth,-0.5*\morphismheight-\savedlengthsw)
                to (-0.4*\overallwidth,-0.5*\morphismheight);
        \fi
        \ifconnectser
            \draw [reverse arrow style] (0.4*\overallwidth,-0.5*\morphismheight-\savedlengthse)
                to (0.4*\overallwidth,-0.5*\morphismheight);
        \fi
        \ifconnectnr
            \draw [reverse arrow style] (0,0.5*\morphismheight)
                to (0,0.5*\morphismheight+\savedlengthn);
        \fi
        \ifconnectsr
            \draw [reverse arrow style] (0,-0.5*\morphismheight-\savedlengths)
                to (0,-0.5*\morphismheight);
        \fi
    \end{scope}
    }
}
\tikzset{dashedmorphism/.style={dashedmorphismshape, node on layer=foreground}}
\pgfdeclareshape{swish right}
{
    \savedanchor\centerpoint
    {
        \pgf@x=0pt
        \pgf@y=0pt
    }
    \anchor{center}{\centerpoint}
    \anchorborder{\centerpoint}
    \anchor{north}
    {
        \pgf@x=\minimummorphismwidth
        \divide\pgf@x by 5
        \pgf@y=\morphismheight
        \divide\pgf@y by 2
        \advance\pgf@y by \connectheight
    }
    \anchor{south}
    {
        \pgf@x=-\minimummorphismwidth
        \divide\pgf@x by 5
        \pgf@y=-\morphismheight
        \divide\pgf@y by 2
        \advance\pgf@y by -\connectheight
    }
    \backgroundpath
    {
        \pgfsetstrokecolor{black}
        \pgfsetlinewidth{\thickness}
        \pgfpathmoveto{\pgfpoint
            {-0.2*\minimummorphismwidth}
            {-0.5*\morphismheight-\connectheight}}
        \pgfpathcurveto
            {\pgfpoint{-0.2*\minimummorphismwidth}{0pt}}
            {\pgfpoint{0.2*\minimummorphismwidth}{0pt}}
            {\pgfpoint
                {0.2*\minimummorphismwidth}
                {0.5*\morphismheight+\connectheight}}
        \pgfusepath{stroke}
    }
}
\pgfdeclareshape{swish left}
{
    \savedanchor\centerpoint
    {
        \pgf@x=0pt
        \pgf@y=0pt
    }
    \anchor{center}{\centerpoint}
    \anchorborder{\centerpoint}
    \anchor{north}
    {
        \pgf@x=-\minimummorphismwidth
        \divide\pgf@x by 5
        \pgf@y=\morphismheight
        \divide\pgf@y by 2
        \advance\pgf@y by \connectheight
    }
    \anchor{south}
    {
        \pgf@x=\minimummorphismwidth
        \divide\pgf@x by 5
        \pgf@y=-\morphismheight
        \divide\pgf@y by 2
        \advance\pgf@y by -\connectheight
    }
    \backgroundpath
    {
        \pgfsetstrokecolor{black}
        \pgfsetlinewidth{\thickness}
        \pgfpathmoveto{\pgfpoint
            {0.2*\minimummorphismwidth}
            {-0.5*\morphismheight-\connectheight}}
        \pgfpathcurveto
            {\pgfpoint{0.2*\minimummorphismwidth}{0pt}}
            {\pgfpoint{-0.2*\minimummorphismwidth}{0pt}}
            {\pgfpoint
                {-0.2*\minimummorphismwidth}
                {0.5*\morphismheight+\connectheight}}
        \pgfusepath{stroke}
    }
}
\pgfdeclareshape{state}
{
    \savedanchor\centerpoint
    {
        \pgf@x=0pt
        \pgf@y=0pt
    }
    \anchor{center}{\centerpoint}
    \anchorborder{\centerpoint}
    \saveddimen\overallwidth
    {
        \pgf@x=3\wd\pgfnodeparttextbox
        \ifdim\pgf@x<\minimumstatewidth
            \pgf@x=\minimumstatewidth
        \fi
    }
    \savedanchor{\upperrightcorner}
    {
        \pgf@x=.5\wd\pgfnodeparttextbox
        \pgf@y=.5\ht\pgfnodeparttextbox
        \advance\pgf@y by -.5\dp\pgfnodeparttextbox
    }
    \anchor{north}
    {
        \pgf@x=0pt
        \pgf@y=0pt
    }
    \anchor{south}
    {
        \pgf@x=0pt
        \pgf@y=0pt
    }
    \anchor{A}
    {
        \pgf@x=-\overallwidth
        \divide\pgf@x by 4
        \pgf@y=0pt
    }
    \anchor{B}
    {
        \pgf@x=\overallwidth
        \divide\pgf@x by 4
        \pgf@y=0pt
    }
    \anchor{text}
    {
        \upperrightcorner
        \pgf@x=-\pgf@x
        \ifhflip
            \pgf@y=-\pgf@y
            \advance\pgf@y by 0.4\stateheight
        \else
            \pgf@y=-\pgf@y
            \advance\pgf@y by -0.4\stateheight
        \fi
    }
    \backgroundpath
    {
       \begin{pgfonlayer}{foreground}
        \pgfsetstrokecolor{black}
        \pgfsetlinewidth{\thickness}
        \pgfpathmoveto{\pgfpoint{-0.5*\overallwidth}{0}}
        \pgfpathlineto{\pgfpoint{0.5*\overallwidth}{0}}
        \ifhflip
            \pgfpathlineto{\pgfpoint{0}{\stateheight}}
        \else
            \pgfpathlineto{\pgfpoint{0}{-\stateheight}}
        \fi
        \pgfpathclose
        \ifblack
            \pgfsetfillcolor{black!50}
            \pgfusepath{fill,stroke}
        \else
            \pgfusepath{stroke}
        \fi
       \end{pgfonlayer}
    }
}
\makeatother

\newcommand{\tinyhandle}[1][dot]{\ensuremath{\smash{\raisebox{-1pt}{\ensuremath{\hspace{-2pt}\begin{pic}[scale=0.33]
        \node (0) at (0,0) {};
        \node[dot, inner sep=1.0pt] (1) at (0,0.3) {};
        \node[dot, inner sep=1.0pt] (2) at (0,0.7) {};
        \node (3) at (0,1) {};
        \draw (0.center) to (1.center);
        \draw (2.center) to (3.center);
        \draw[in=180, out=180, looseness=2] (1.center) to (2.center);
        \draw[in=0, out=0, looseness=2] (1.center) to (2.center);
\end{pic}\hspace{-1pt}}}}}}

\newcommand{\tinycup}{\begin{pic}[scale=0.17]
    \draw (0,0) to[out=-90,in=-90,looseness=2] (1.5,0);
\end{pic}}

\makeatletter
\pgfdeclareshape{arrow shape}
{
    \savedanchor\centerpoint
    {
        \pgf@x=0pt
        \pgf@y=0pt
    }
    \anchor{center}{\centerpoint}
    \anchor{south}{\centerpoint}
    \anchor{north}{\centerpoint}
    \anchorborder{\centerpoint}
    \backgroundpath
    {
        \begin{pgfonlayer}{foreground}
        \pgfsetstrokecolor{black}
        \pgfsetarrowsstart{>[\strarr]}
        \pgfsetlinewidth{\thickness}
        \pgfpathmoveto{\pgfpoint{0}{-2pt}}
        \pgfpathlineto{\pgfpoint{0}{1pt}}
        \pgfusepath{stroke}
        \end{pgfonlayer}
    }
}
\pgfdeclareshape{double arrow shape}
{
    \savedanchor\centerpoint
    {
        \pgf@x=0pt
        \pgf@y=0pt
    }
    \anchor{center}{\centerpoint}
    \anchor{south}{\centerpoint}
    \anchor{north}{\centerpoint}
    \anchorborder{\centerpoint}
    \backgroundpath
    {
        \begin{pgfonlayer}{foreground}
        \pgfsetstrokecolor{black}
        \pgfsetarrowsstart{>[\strarr]}
        \pgfsetlinewidth{\thickness}
        \pgfpathmoveto{\pgfpoint{0}{-3.5pt}}
        \pgfpathlineto{\pgfpoint{0}{-0.5pt}}
        \pgfusepath{stroke}
        \pgfpathmoveto{\pgfpoint{0}{-0.5pt}}
        \pgfpathlineto{\pgfpoint{0}{2.5pt}}
        \pgfusepath{stroke}
        \end{pgfonlayer}
    }
}
\pgfdeclareshape{reverse arrow shape}
{
    \savedanchor\centerpoint
    {
        \pgf@x=0pt
        \pgf@y=0pt
    }
    \anchor{center}{\centerpoint}
    \anchor{south}{\centerpoint}
    \anchor{north}{\centerpoint}
    \anchorborder{\centerpoint}
    \backgroundpath
    {
        \begin{pgfonlayer}{foreground}
        \pgfsetstrokecolor{black}
        \pgfsetarrowsstart{<[\strarr]}
        \pgfsetlinewidth{\thickness}
        \pgfsetlinewidth{\thickness}
        \pgfpathmoveto{\pgfpoint{0}{-2pt}}
        \pgfpathlineto{\pgfpoint{0}{-1pt}}
        \pgfusepath{stroke}
        \end{pgfonlayer}
    }
}
\pgfdeclareshape{reverse double arrow shape}
{
    \savedanchor\centerpoint
    {
        \pgf@x=0pt
        \pgf@y=0pt
    }
    \anchor{center}{\centerpoint}
    \anchor{south}{\centerpoint}
    \anchor{north}{\centerpoint}
    \anchorborder{\centerpoint}
    \backgroundpath
    {
        \begin{pgfonlayer}{foreground}
        \pgfsetstrokecolor{black}
        \pgfsetarrowsstart{<[\strarr]}
        \pgfsetlinewidth{\thickness}
        \pgfpathmoveto{\pgfpoint{0}{-3.5pt}}
        \pgfpathlineto{\pgfpoint{0}{-0.5pt}}
        \pgfusepath{stroke}
        \pgfpathmoveto{\pgfpoint{0}{-0.5pt}}
        \pgfpathlineto{\pgfpoint{0}{2.5pt}}
        \pgfusepath{stroke}
        \end{pgfonlayer}
    }
}
\makeatother

\tikzset{arrow/.pic={
    \path[pic actions, decoration={ markings,
      mark=at position 1 with {\arrow{>[\strarr]}}
    },
    postaction={decorate}] (0,1pt) -- (0,2pt);
},double arrow/.pic={
    \path[pic actions, decoration={ markings,
      mark=at position \pgfdecoratedpathlength-3pt with {\arrow{>[\strarr]}},
      mark=at position \pgfdecoratedpathlength with {\arrow{>[\strarr]}}
    },
    postaction={decorate}] (0,-6pt) -- (0,4pt);
},
  reverse arrow/.pic={
    \path[pic actions, decoration={ markings,
      mark=at position 1 with {\arrow{<[\strarr]}}
    },
    postaction={decorate}] (0,1pt) -- (0,2pt);
}
}

\tikzset{surface picture/.style={xscale={0.6}, yscale=0.8, line width=\thickness}}
\tikzset{three dimensional picture/.style={}}
\tikzset{morphismtwocell/.style={morphism, width=0.5cm, connect ne length=0cm, connect se length=0cm, connect nw length=0cm, connect sw length=0cm}}
\tikzset{nmorphismtwocell/.style={draw, minimum width=0.8cm, connect ne length=0cm, connect se length=0cm, connect nw length=0cm, connect sw length=0cm}}

\tikzset{dashback/.style={black!40}}

\tikzset{shade 1 local/.style={shade 1}}
\tikzset{shade 2 local/.style={shade 2}}

\begin{document}
\title{Categorical composable cryptography\thanks{This work was supported by the Air Force Office of Scientific Research under award number FA9550-20-1-0375, Canada’s NFRF and NSERC, an Ontario ERA, and the University of Ottawa’s Research Chairs program.}}
%
%
\author{Anne Broadbent
\orcidID{0000-0003-1911-0093} \and
Martti Karvonen(\Envelope)
\orcidID{0000-0002-8919-343X}}
\authorrunning{A. Broadbent and M. Karvonen}
%
\institute{Department of Mathematics and Statistics, University of Ottawa, Ottawa, Canada
\email{\{abroadbe,martti.karvonen\}@uottawa.ca}}
\maketitle 
\begin{abstract}
We formalize the simulation paradigm of cryptography in terms of category theory and show that protocols secure against abstract attacks form a symmetric monoidal category, thus giving an abstract model of composable security definitions in cryptography. Our model is able to incorporate computational security, set-up assumptions and various attack models such as colluding or independently acting subsets of adversaries in a modular, flexible fashion. We conclude by using string diagrams to rederive the security of the one-time pad and no-go results concerning the limits of bipartite and tripartite cryptography, ruling out e.g., composable commitments and broadcasting.

\keywords{Cryptography  \and composable security \and quantum cryptography \and category theory}
\end{abstract}
\section{Introduction}
Modern cryptographic protocols are complicated algorithmic entities, and their security analyses are often no simpler than the protocols themselves. Given this complexity, it would be highly desirable to be able to design protocols and reason about them compositionally, \ie by breaking them down into smaller constituent parts. In particular, one would hope that combining protocols proven secure results in a secure protocol without need for further security proofs. However, this is not the case for stand-alone security notions that are common in cryptography. To illustrate such failures of composability, let us consider the history of quantum key distribution (QKD), as recounted in~\cite{PR14arxiv}: QKD was originally proposed in the 80s~\cite{BB84}. The first security proofs against unbounded adversaries followed a decade later~\cite{May96,BBB+00,SP00,May01}. However, since composability was originally not a concern, it was later realized that the original security definitions did not provide a good enough level of security~\cite{KRBM07}---they didn't guarantee security if the keys were to be actually used, since even a partial leak of the key would compromise the rest. The story ends on a positive note, as eventually a new security criterion was proposed, together with stronger proofs~\cite{Ren05,BHL+05}.

In this work we initiate a categorical study of composable security definitions in cryptography. In the viewpoint developed here one thinks of cryptography as a resource theory: cryptographic functionalities (e.g.~secure communication channels) are viewed as resources and cryptographic protocols let one transform some starting resources to others. For instance, one can view the one-time-pad as a protocol that transforms an authenticated channel and a shared secret key into a secure channel. For a given protocol, one can then study whether it is secure against some (set of) attack model(s), and protocols secure against a fixed set of models can always be composed sequentially and in parallel.

This is in fact the viewpoint taken in constructive cryptography~\cite{Mau11}, which also develops the one-time-pad example above in more detail. However~\cite{Mau11} does not make a formal connection to resource theories as usually understood, whether as in quantum physics~\cite{horodecki:resource,chitambar:resource}, or more generally as defined in order theoretic~\cite{Fritz2015} or categorical~\cite{CFS16} terms. Instead, constructive cryptography is usually combined with abstract cryptography~\cite{MR11} which is formalized in terms of a novel algebraic theory of systems~\cite{MMP+18}.

Our work can be seen as a particular formalization of the ideas behind constructive cryptography, or alternatively as giving a categorical account of the real-world-ideal-world paradigm (also known as the simulation paradigm~\cite{GM84}), which underlies more concrete frameworks for composable security, such as universally composable cryptography~\cite{Can01} and others~\cite{PW00,BPW04,BPW07,MT13,HS15,LHM19,KTR20}. We will discuss these approaches and abstract and constructive cryptography in more detail in Section~\ref{sec:relatedwork}

Our long-term goal is to enable cryptographers to reason about composable security at the same level of formality as stand-alone security, \emph{without having to fix all the details of a machine model nor having to master category theory}. Indeed, our current results already let one define multipartite protocols and security against arbitrary subsets of malicious adversaries \emph{in any symmetric monoidal category $\CC$}. Thus, as long as one's model of interactive computation results in a symmetric monoidal category, or more informally, one is willing to use pictures such as \cref{fig:attackonprod} to depict connections between computational processes without further specifying the order in which the picture was drawn, one can use the simulation paradigm to reason about multipartite security against malicious participants composably---and specifying finer details of the computational model is only needed to the extent that it affects the validity of one's argument. Moreover, as our attack models and composition theorems are fairly general, we hope that more refined models of adversaries can be incorporated.

We now highlight our contributions to cryptography:
We show how to adapt resource theories as categorically formulated~\cite{CFS16} in order to reason abstractly about \emph{secure} transformations between resources. This is done in Section~\ref{sec:crypto} by formalizing the simulation paradigm in terms of an abstract attack model (Definition~\ref{def:attack}), designed to be general enough to capture standard attack models of interest (and more) while still structured enough to guarantee composability. This section culminates in Corollary~\ref{cor:simultaneoussafety}, which shows that for any fixed set of attack models, the class of protocols secure against each of them results in a symmetric monoidal category. In Theorem~\ref{thm:perfectlifting} we observe that under suitable conditions, images of secure protocols under monoidal functors remain secure, which gives an abstract variant of the lifting theorem~\cite[Theorem 15]{Unr10} that states that perfectly UC-secure protocols are quantum UC-secure.
We adapt this framework to model \emph{computational security} in two ways: either by replacing equations with an equivalence relation, abstracting the idea of computational indistinguishability, as is done in section~\ref{sec:extensions}, or by working with a notion of distance, deferred to a full version. In the case of a distance, one can then either explicitly bound the distance between desired and actually achieved behavior, or work with sequences of protocols that converge to the target in the limit: the former models working in the finite-key regimen~\cite{TLGR12} and the latter models the kinds of asymptotic security and complexity statements that are common in cryptography.
Finally, we apply the framework developed to study bipartite and tripartite cryptography. We first prove pictorially the security of the one-time pad. We then reprove the no-go-theorems of~\cite{PR08,MR11,MMP+18} concerning two-party commitments (resp. three-party broadcasting) in this setting, and reinterpret them as limits on what can be achieved securely in any compact closed category (resp. symmetric monoidal category). The key steps of the proof are done graphically, thus opening the door for cryptographers to use such pictorial representations as rigorous tools rather than merely as illustrations.

Moreover, we discuss some categorical constructions capturing aspects of resource theories appearing in the physics literature. These contributions may be of independent interest for further categorical studies on resource theories. 
 In~\cite{CFS16} it is observed that many resource theories arise from an inclusion $\cat{C}_F\hookrightarrow\cat{C}$ of free transformations into a larger monoidal category, by taking the resource theory of states. We observe that this amounts to applying the monoidal Grothendieck construction~\cite{moeller:monoidalgrothendieck} to the functor $\CF\to\CC\xrightarrow{\hom(I,-)}\cat{Set}$. This suggests applying this construction more generally to the composite of monoidal functors $F\colon\cat{D}\to\cat{C}$ and  $R\colon \cat{C}\to\Set$.
        In Example~\ref{ex:n-partite} we note that choosing $F$ to be the $n$-fold monoidal product $\CC^n\to\CC$ captures resources shared by $n$ parties and $n$-partite transformations between them.
        In the extended version, we model categorically situations where there is a notion of distance between resources, and instead of exact resource conversions one either studies approximate transformations or sequences of transformations that succeed in the limit.
        In the extended version, we discuss a variant of a construction on monoidal categories, used in special cases in~\cite{fongetal:backprop} and discussed in more detail in~\cite{cruttwell2021categorical,Gavranovic:compositional}, that allows one to declare some resources free and thus enlarge the set of possible resource conversions.
\subsection{Related work}\label{sec:relatedwork}
We have already mentioned that cryptographers have developed a plethora of frameworks for composable security, such as universally composable cryptography~\cite{Can01}, reactive simulatability~\cite{PW00,BPW04,BPW07} and others~\cite{MT13,HS15,LHM19,KTR20}. Moreover, some of these frameworks have been adapted to the quantum setting~\cite{BM04arxiv,Unr10,MR09}. One might hence be tempted to think that the problem of composability in cryptography has been solved. However, it is fair to say that most mainstream cryptography is not formulated composably and that composable cryptography has yet to realize its full potential. Moreover, this proliferation of frameworks should be taken as evidence of the continued importance of the issue, and is in fact reflected by the existence of a recent Dagstuhl seminar on this matter~\cite{CKL+19}. Indeed, the aforementioned frameworks mostly consist of setting up fairly detailed models of interacting machines, which as an approach suffers from two drawbacks:
    Firstly, in order to be more realistic, the detailed models are often complicated, both to reason in terms of and to define, thus making practicing cryptographers less willing to use them. Perhaps more importantly it is not always clear whether the results proven in a particular model apply more generally for other kinds of machines, whether those of a competing framework or those in the real world. It is true that the choice of a concrete machine model does affect what can be securely achieved---for instance, quantum cryptography differs from classical cryptography and similarly classical cryptography behaves differently in synchronous and asynchronous settings~\cite{BCG93,KMTZ13}. Nevertheless, one might hope that composable cryptography could be done at a similar level of formality as complexity theory, where one rarely worries about the number of tapes in a Turing machine or of other low-level details of machine models.
    Second, changing the model slightly (to \eg model different kinds of adversaries or to incorporate a different notion of efficiency) often requires reproving ``composition theorems'' of the framework or at least checking that the existing proof is not broken by the modification.

In contrast to frameworks based on detailed machine models, there are two closely related top-down approaches to cryptography: constructive cryptography~\cite{Mau11} and its cousin abstract cryptography~\cite{MR11}. We are indebted to both of these approaches, and indeed our framework could be seen as formalizing the key idea of constructive cryptography---namely, cryptography as a resource theory---and thus occupying a similar space as abstract cryptography. A key difference is that constructive cryptography is usually instantiated in terms of abstract cryptography~\cite{MR11}, which in turn is based on a novel algebraic theory of systems~\cite{MMP+18}. However, our work is not merely a translation from this theory to categorical language, as there are important differences and benefits that stem from formalizing cryptography in terms of a well-established and well-studied algebraic theory of systems---that of (symmetric) monoidal categories:

        The fact that cryptographers wish to compose their protocols \emph{sequentially and in parallel} strongly suggests using \emph{monoidal categories}, that have these composition operations as primitives. In our framework, protocols secure against a fixed set of attack models results in a symmetric monoidal category. In contrast, the algebraic theory of systems~\cite{MMP+18} on which abstract cryptography is based takes parallel composition and internal wiring as its primitives. This design choice results in some technical kinks and tangles that are natural with any novel theory but have already been smoothed out in the case of category theory. For instance, in the algebraic theory of systems of~\cite{MMP+18} the parallel composition is a partial operation and in particular the parallel composite of a system with itself is never defined\footnote{While the suggested fix is to assume that one has ``copies'' of the same system with disjoint wire labels, it is unclear how one recognizes or even defines \emph{in terms of the system algebra} that two distinct systems are copies of each other.} and the set of wires coming out of a system is fixed once and for all\footnote{Indeed, while~\cite{portmann:causal} manages to bundle and unbundle ports along isomorphism when convenient, it seems like the chosen technical foundation makes this more of a struggle than it should be.}. In contrast, in a monoidal category parallel composition is a total operation and whether one draws a box with $n$ output wires of types $A_1,\dots A_n$ or single output wire of type $\bigotimes_{i=1}^n A_i$ is a matter of convenience. Technical differences such as these make a direct formal comparison or translation between the frameworks difficult, even if informally and superficially there are similarities.

        We do not abstract away from an attacker model, but rather make it an explicit part of the formalism that can be modified without worrying about composability. This makes it possible to consider and combine very easily different security properties, and in particular paves the way to model attackers with limited powers such as honest-but-curious adversaries. In our framework, one can first fix a protocol transforming some resource to another one, and then discuss whether this transformation is secure against different attack models. In contrast, in abstract cryptography a cryptographic resource is a tuple of functionalities, one for each set of dishonest parties, and thus has no prior existence before fixing the attack model. This makes the question ``what attack models is this protocol secure against?'' difficult to formalize.

        As category theory is de facto the lingua franca between several subfields of mathematics and computer science, elucidating the categorical structures present in cryptography opens up the door to further connections between cryptography and other fields. For instance, game semantics readily gives models of interactive, asynchronous and probabilistic (or quantum) computation~\cite{winskel:game,clairambault:gamesforquantum,clairambaultetal:gamesforquantum2} in which our theory can be instantiated, and thus further paves the way for programming language theory to inform cryptographic models of concurrency.

        Category theory comes with existing theory, results and tools that can readily be applied to questions of cryptographic interest. In particular, the graphical calculi of symmetric monoidal and compact closed categories~\cite{Sel10} enables one to rederive impossibility results shown in~\cite{PR08,MR11,MMP+18} purely pictorially. In fact, such pictures were already often used as heuristic devices that illuminate the official proofs, and viewing these pictures categorically lets us promote them from mere illustrations to rigorous yet intuitive proofs. Indeed, in~\cite[Footnote 27]{MR11} the authors suggest moving from a 1-dimensional symbolic presentation to a 2-dimensional one, and this is exactly what the graphical calculus achieves.

The approaches above result in a framework where security is defined so as to guarantee composability. In contrast, approaches based on various protocol logics~\cite{DMP01,DMP03,DDMP03b,DDMP03a,DDMP05,DDMR07} aim to characterize situations where composition can be done securely, even if one does not use composable security definitions throughout. As these approaches are based on process calculi, they are categorical under the hood~\cite{pavlovic1997,Mifsudetal:controlstructures} even if not overtly so. There is also earlier work explicitly discussing category theory in the context of cryptography~\cite{breiner:graphicaldicrypto,coecke:graphicalqkd,sunwang:graphicalbc,breiner:selftesting,heunen:qkd,hillebrand:superdense,coecke:environment,kissinger2017picture,stay:crypto,dusko:crypto,hines:diagrammaticrypto,pavlovic2012tracing}, but they concern stand-alone security of particular cryptographic protocols, rather than categorical aspects of composable security definitions.
\section{Resource theories}
We briefly review the categorical viewpoint on resource theories of~\cite{CFS16}. Roughly speaking, a resource theory can be seen as a SMC but the change in terminology corresponds to a change in viewpoint: usually in category theory one studies global properties of a category, such as the existence of (co)limits, relationships to other categories, etc. In contrast, when one views a particular SMC $\CC$ as resource theory, one is interested in local questions. One thinks of objects of $\CC$ as resources, and morphisms as processes that transform a resource to another. From this point of view, one mostly wishes to understand whether $\hom_\CC(X,Y)$ is empty or not for resources $X$ and~$Y$ of interest. Thus from the resource-theoretic point of view, most of the interesting information in $\CC$ is already present in its preorder collapse. As concrete examples of resource-theoretic questions, one might wonder if
    (i) some noisy channels can simulate a (almost) noiseless channel~\cite[Example 3.13.]{CFS16},
    (ii) there is a protocol that uses only local quantum operations and classical communication and transforms a particular quantum state to another one~\cite{Chitambaretal:LOCC},
    (iii) some non-classical statistical behavior can be used to simulate other such behavior~\cite{abramskyetal:comonadicview}.
In~\cite{CFS16} the authors show how many familiar resource theories arise in a uniform fashion: starting from an SMC $\CC$ of processes equipped with a wide sub-SMC $\CF$, the morphisms of which correspond to ``free'' processes, they build several resource theories (=SMCs). Perhaps the most important of these constructions is the resource theory of states: given $\CF\hookrightarrow\CC$, the corresponding resource theory of states can be explicitly constructed by taking the objects of this resource theory to be states of $\CC$, \ie maps $r\colon I\to A$ for some $A$, and maps $r\to s$ are maps $f\colon A\to B$ in $\CF$ that transform $r$ to $s$ as in \cref{fig:state_transform}.

We now turn our attention towards cryptography. As contemporary cryptography is both broad and complex in scope, any faithful model of it is likely to be complicated as well. A benefit of the categorical idiom is that we can build up to more complicated models in stages, which is what we will do in the sequel. We phrase our constructions in terms of an arbitrary SMC $\CC$, but in order to model actual cryptographic protocols, the morphisms of $\CC$ should represent interactive computational machines with open ``ports'', with composition then amounting to connecting such machines together. Different choices of $\CC$ set the background for different kinds of cryptography, so that quantum cryptographers want $\CC$ to include quantum systems whereas in classical cryptography it is sufficient that these computational machines are probabilistic. Constructing such categories $\CC$ in detail is not trivial but is outside our scope---we will discuss this in more detail in section~\ref{sec:outlook}.

Our first observation is that there is no reason to restrict to inclusions $\CF\hookrightarrow\CC$ in order to construct a resource theory of states. Indeed, while it is straightforward to verify explicitly that the resource theory of states is a symmetric monoidal category, it is instructive to understand more abstractly why this is so: in effect, the constructed category is the category of elements of the composite functor $\CF\to\CC\xrightarrow{\hom(I,-)}\cat{Set}$. As this composite is a (lax) symmetric monoidal functor, the resulting category is automatically symmetric monoidal as observed in~\cite{moeller:monoidalgrothendieck}. Thus this construction goes through for any symmetric (lax) monoidal functors $\cat{D}\xrightarrow{F}\cat{C}\xrightarrow{R}\Set$. Here we may think of $F$ as interpreting free processes into an ambient category of all processes, and $R\colon\CC\to\cat{Set}$ as an operation that gives for each object $A$ of $\CC$ the set $R(A)$ of resources of type~$A$.

Explicitly, given symmetric monoidal functors $\cat{D}\xrightarrow{F}\cat{C}\xrightarrow{R}\Set$, the category of elements $\res RF$ has as its objects pairs $(r,A)$ where $A$ is an object of $\cat{D}$ and $r\in RF(A)$, the intuition being that $r$ is a resource of type $F(A)$. A morphism $(r,A)\to (s,B)$ is given by a morphism $f\colon A\to B$ in $\cat{D}$ that takes $r$ to $s$, \ie satisfies $RF(f)(r)=s$. The symmetric monoidal structure comes from the symmetric monoidal structures of $\cat{D}, \Set$ and $RF$. Somewhat more explicitly, $(r,A)\otimes (s,B)$ is defined by $(r\otimes s,A\otimes B)$ where $r\otimes s$ is the image of $(r,s)$ under the function $RF(A)\times RF(B)\to RF(A\otimes B)$ that is part of the monoidal structure on $RF$, and on morphisms of $\res RF$ the monoidal product is defined from that of $\cat{D}$.

From now on we will assume that $F$ is strong monoidal, and while $R=\hom(I,-)$ captures our main examples of interest, we will phrase our results for an arbitrary lax monoidal $R$. This relaxation allows us to capture the $n$-partite structure often used when studying cryptography, as shown next.
\begin{example}\label{ex:n-partite}
Consider the resource theory induced by $\CC^n\xrightarrow{\otimes}\CC\xrightarrow{\hom(I,-)}\Set$, where we write $\otimes$ for the $n$-fold monoidal product\footnote{As $\CC$ is symmetric, the functor $\otimes$ is strong monoidal.}. The resulting resource theory has a natural interpretation in terms of $n$ agents trying to transform resources to others: an object of this resource theory corresponds to a pair $((A_i)_{i=1}^n,r\colon I\to \bigotimes A_i)$, and can be thought of as an $n$-partite state, depicted in \cref{fig:n-partite_state}, where the $i$th agent has access to a port of type $A_i$. A morphism $\f=(f_1,\dots f_n)\colon ((A_i)_{i=1}^n,r)\to ((B_i)^n_{i=1},s)$ between such resources then amounts to a protocol that prescribes, for each agent $i$ a process $f_i$ that they should perform so that~$r$ gets transformed to $s$ as in \cref{fig:n-partite_state_transformation}.
\end{example}
\begin{figure}
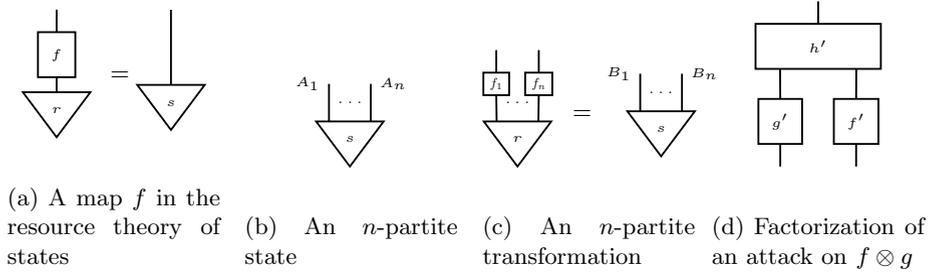

    \centering
\begin{subfigure}[b]{0.23\textwidth}
\[\begin{pic}
    \node[morphism] (f) at (0,.5) {$f$};
    \node[state] (x) at (0,0) {$r$};
    \draw (f.south) to node[right] {$$} (x.north);
    \draw (f.north) to +(0,.3) node[right] {};
  \end{pic}\enspace=\begin{pic}
  \node[state] (y) at (0,0) {$s$};
  \draw (y.north) to +(0,1);
\end{pic}\]
	\caption{A map $f$ in the resource theory of states}
	\label{fig:state_transform}
	\end{subfigure}
	~
    \begin{subfigure}[b]{0.23\textwidth}
\[
    \begin{pic}
    \node[state] (f) at (0,0) {$s$};
    \node (a) at (0,.25) {$\ldots$};
    \draw ([xshift=-1.5pt]f.A) to +(0,.5) node[left] {$A_1$};
    \draw ([xshift=1.5pt]f.B) to +(0,.5) node[right] {$A_n$};
  \end{pic}
  \]
        \caption{An $n$-partite state}
        \label{fig:n-partite_state}
    \end{subfigure}
    ~
    \begin{subfigure}[b]{0.23\textwidth}
\[
  \begin{pic}
    \node[state] (r) at (0,0) {$r$};
    \node (a) at (0,.25) {$\ldots$};
    \node[morphism,scale=.5,font=\normalsize] (f) at (-.28,.5) {$f_1$};
    \node[morphism,scale=.5,font=\normalsize] (g) at (.28,.5) {$f_n$};
    \draw ([xshift=-1.5pt]r.A) to (f.south);
    \draw ([xshift=1.5pt]r.B) to (g.south);
    \draw (f.north) to +(0,.3);
    \draw (g.north) to +(0,.3);
  \end{pic}\enspace=
    \begin{pic}
    \node[state] (f) at (0,0) {$s$};
    \node (a) at (0,.25) {$\ldots$};
    \draw ([xshift=-1.5pt]f.A) to +(0,.5) node[left] {$B_1$};
    \draw ([xshift=1.5pt]f.B) to +(0,.5) node[right] {$B_n$};
  \end{pic}
  \]
        \caption{An $n$-partite transformation}
        \label{fig:n-partite_state_transformation}
    \end{subfigure}
   ~\begin{subfigure}[b]{0.23\textwidth}
   \[\begin{pic}
	\node[morphism] (g) at (0,0) {$g'$};
	\node[morphism] (f) at (1,0) {$f'$};
	\setlength\minimummorphismwidth{13mm}
	\node[morphism] (h) at (0.5,1) {$h'$};
	\draw (g.north) to ([xshift=0.5pt]h.south west);
	\draw (f.north) to ([xshift=-0.5pt]h.south east);
	\draw (h.north) to +(0,.3);
	\draw (g.south) to +(0,.-.3);
	\draw (f.south) to +(0,.-.3);
\end{pic}
\]
        \caption{Factorization of an attack on $f\otimes g$}
        \label{fig:attackonprod}
   \end{subfigure}
\caption{Some resource transformations\vspace{-1em}}
\end{figure}
In this resource theory, all of the agents are equally powerful and can perform all processes allowed by $\CC$, and this might be unrealistic: first of all, $\CC$ might include computational processes that are too powerful/expensive for us to use in our cryptographic protocols. Moreover, having agents with different computational powers is important to model \eg blind quantum computing~\cite{BFK09} where a client with access only to limited, if any, quantum computation tries to securely delegate computations to a server with a powerful quantum computer. This limitation is easily remedied: we could take the $i$th agent to be able to implement computations in some sub-SMC $\CC_i$ of $\CC$, and then consider $\prod_{i=1}^n \CC_i\to\CC$.

A more serious limitation is that such transformations have no security guarantees---they only work if each agent performs~$f_i$ as prescribed by the protocol. We fix this next.

\section{Cryptography as a resource theory}\label{sec:crypto}
\begin{figure}
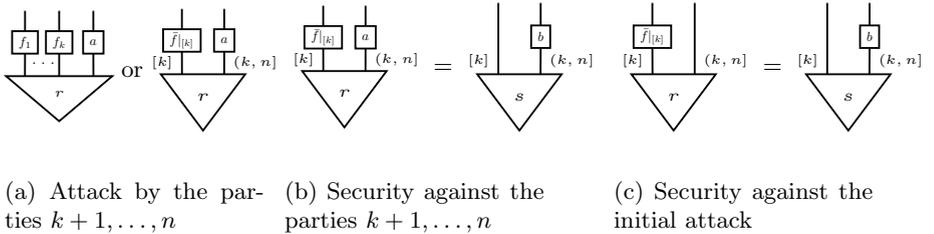
\vspace{-7.5mm}
    \centering
\begin{subfigure}[b]{0.28\textwidth}
\[\begin{pic}[scale=0.9]
  	\setlength\minimumstatewidth{14mm}
    \node[state] (r) at (0,0) {$r$};
    \node (a) at (-.23,.2) {$\ldots$};
    \node[morphism,scale=.5,font=\normalsize] (f) at (-.5,.5) {$f_1$};
    \node[morphism,scale=.5,font=\normalsize] (h) at (0,.5) {$f_{k}$};
    \node[morphism,scale=.5,font=\normalsize] (g) at (.5,.5) {$a$};
    \draw ([xshift=-3.2pt]r.A) to (f.south);
    \draw ([xshift=3.2pt]r.B) to (g.south);
    \draw (f.north) to +(0,.3);
    \draw (g.north) to +(0,.3);
     \draw (h.north) to +(0,.3);
    \draw (r.center) to (h.south);
  \end{pic}\text{ or}
\begin{pic}[scale=0.9]
    \node[morphism,scale=.5,font=\normalsize] (f) at (-.31,.5) {$\f|_{[k]}$};
     \node[morphism,scale=.5,font=\normalsize] (g) at (.31,.5) {$a$};
    \node[state,scale=1.25] (x) at (0,0) {$r$};
    \draw (f.south) to node[left] {$[k]$} (x.A);
    \draw (f.north) to +(0,.3) node[right] {};
    \draw (g.south) to node[right] {$(k,n]$} (x.B);
    \draw (g.north) to +(0,.3) node[right] {};
  \end{pic}\]
	\caption{Attack by the parties $k+1,\dots ,n$}
	\label{fig:attack}
	\end{subfigure}
	\enspace
    \begin{subfigure}[b]{0.28\textwidth}
\[\begin{pic}[scale=0.9]
    \node[morphism,scale=.5,font=\normalsize] (f) at (-.31,.5) {$\f|_{[k]}$};
     \node[morphism,scale=.5,font=\normalsize] (g) at (.31,.5) {$a$};
    \node[state,scale=1.25] (x) at (0,0) {$r$};
    \draw (f.south) to node[left] {$[k]$} (x.A);
    \draw (f.north) to +(0,.3) node[right] {};
    \draw (g.south) to node[right] {$(k,n]$} (x.B);
    \draw (g.north) to +(0,.3) node[right] {};
  \end{pic}=\begin{pic}
     \node[morphism,scale=.5,font=\normalsize] (g) at (.28,.5) {$b$};
    \node[state,scale=1.25] (x) at (0,0) {$s$};
    \draw (x.A) to node[left] {$[k]$} +(0,.35) to +(0,.95);
    \draw (g.south) to node[right] {$(k,n]$} (x.B);
    \draw (g.north) to +(0,.3) node[right] {};
  \end{pic}\]
        \caption{Security against the parties $k+1,\dots ,n$}
        \label{fig:security}
    \end{subfigure}
    \qquad\enspace
    \begin{subfigure}[b]{0.28\textwidth}
\[\begin{pic}
    \node[morphism,scale=.5,font=\normalsize] (f) at (-.28,.5) {$\f|_{[k]}$};
    \node[state,scale=1.25] (x) at (0,0) {$r$};
    \draw (f.south) to node[left] {$[k]$} (x.A);
    \draw (f.north) to +(0,.3) node[right] {};
    \draw (x.B) to node[right]  {$(k,n]$} +(0,.35)  to +(0,.95);
  \end{pic}=\begin{pic}
     \node[morphism,scale=.5,font=\normalsize] (g) at (.28,.5) {$b$};
    \node[state,scale=1.25] (x) at (0,0) {$s$};
    \draw (x.A) to node[left] {$[k]$} +(0,.35) to +(0,.95);
    \draw (g.south) to node[right] {$(k,n]$} (x.B);
    \draw (g.north) to +(0,.3) node[right] {};
  \end{pic}\]
        \caption{Security against the initial attack}
        \label{fig:initial_attack}
    \end{subfigure}
\caption{Attacks and security constraints\vspace{-1em}}
\end{figure}
In order for a protocol $\f=(f_1,\dots ,f_n)\colon ((A_i)_{i=1}^n,r)\to ((B_i)^n_{i=1},s)$ to be secure, we should have some guarantees about what happens if, as a result of \emph{an attack} on the protocol, something else than $(f_1,\dots ,f_n)$ happens. For instance, some subset of the parties might deviate from the protocol and do something else instead. In the simulation paradigm~\cite{GM84}, security is then defined by saying that, anything that could happen when running the real protocol, \ie $\f$ with $r$, could also happen in the ideal world, \ie with $s$. A given protocol might be secure against some kinds of attacks and insecure against others, so we define security against an abstract attack model. This abstract notion of an attack model is one of the main definitions of our paper. It isolates conditions needed for the composition theorem (Theorem ~\ref{thm:composition}). It also captures our key examples that we use to illustrate the definition after giving it. Note that most proofs are deferred to an extended version.
\begin{definition}\label{def:attack} An \emph{attack model} $\A$ on an SMC $\CC$ consists of giving for each morphism $f$ of $\CC$ a class $\A(f)$ of morphisms of $\CC$ such that
	\begin{enumerate}[(i)]
		\item $f\in\A(f)$ for every $f$.
		\item For any $f\colon A\to B$ and $g\colon B\to C$ and composable $g'\in \A(g),f'\in \A(f)$ we have $g'\circ f'\in \A(g\circ f)$.
		 Moreover, any $h\in \A(g\circ f)$ factorizes as $g'\circ f'$ with $g'\in \A(g)$ and $f'\in \A(f)$.
		\item For any $f\colon A\to B$, $g\colon C\to D$ in $\cat{C}$ and $f'\in \A(f), g'\in \A(g)$ we have $f'\otimes g'\in \A(f\otimes g)$. Moreover, any $h\in \A(f\otimes g)$ factorizes as $h'\circ (f'\otimes g')$ with $f'\in \A(f)$, $g'\in \A(g)$ and $h'\in \A(\id[B\otimes D])$.
	\end{enumerate}
 Let $f\colon (A,r)\to (B,s)$ define a morphism in the resource theory $\res RF$ induced by $F\colon\cat{D}\to\cat{C}$ and $R\colon \cat{C}\to\cat{Set}$. We say that $f$ is \emph{secure} against an attack model $\A$ on $\cat{C}$ (or $\A$-secure) if for any $f'\in \A(F(f))$ with $\dom(f')=F(A)$ there is $b\in\A(\id[F(B)])$ with $\dom(b)=F(B)$ such that $R(f')r=R(b)s$.
\end{definition}
The above definition of security asks for perfect equality and corresponds to information-theoretic security in cryptography. This is often too much to hope for, and we will replace this by an equivalence relation in section~\ref{sec:extensions} and by a notion of distance in an extended version.

The intuition is that $\A$ gives, for each process in $\CC$, the set of behaviors that the attackers could force to happen instead of honest behavior. In particular, $\A(\id[B])$ give the set of behaviors that is available to attackers given access to a system of type $B$. Then property (i) amounts to the assumption that the adversaries could behave honestly. The first halves of properties (ii) and (iii) say that, given an attack on $g$ and one on $f$, both attacks could happen when composing $g$ and~$f$ sequentially or in parallel. The second parts of these say that attacks on composite processes can be understood as composites of attacks. However, note that (iii) does not say that an attack on a product has to be a product of attacks: the factorization says that any $h\in\A(g\otimes f)$ factorizes as in \cref{fig:attackonprod} with $g'\in \A(g)$, $f'\in \A(f)$ and $h'\in \A(\id[B\otimes D])$. The intuition is that an attacker does not have to attack two parallel protocols independently of each other, but might play the protocols against each other in complicated ways. This intuition also explains why we do not require that all morphisms in $\A(f)$ have $F(A)$ as their domain, despite the definition of $\A$-security quantifying only against those: when factoring $h\in \A(g\circ f)$ as $g'\circ f'$ with $g'\in \A(g)$ and $f'\in \A(f)$, we can no longer guarantee that $F(B)$ is the domain of $g'$---perhaps the attackers take us elsewhere when they perform~$f'$.

If one thinks of $F\colon\cat{D}\to\cat{C}$ as representing the inclusion of free processes into general processes, one also gets an explanation why we do not insist that free processes and attacks live in the same category, \ie that $F=\id[\CC]$. This is simply because we might wish to prove that some protocols are secure against attackers that can use more resources than we wish or can use in the protocols.
\begin{example} For any SMC $\CC$ there are two trivial attack models: the minimal one defined by $\A(f)=\{f\}$ and the maximal one sending $f$ to the class of all morphisms of $\CC$. We interpret the minimal attack model as representing honest behavior, and the maximal one as representing arbitrary malicious behavior.
\end{example}
\begin{proposition}\label{prop:productattackmodel} If $\A_1,\dots ,\A_n$ are attack models on SMCs $\CC_1,\dots ,\CC_n$ respectively, then there is a product $\prod_{i=1}^n\A_i$ attack model on $\prod_{i=1}^n\CC_i$ defined by $(\prod_{i=1}^n \A_i)(f_1,\dots, f_n)=\prod_{i=1}^n \A_i(f_i)$.
\end{proposition}
This proposition, together with the minimal and maximal attack models, is already expressive enough to model multi-party computation where some subset of the parties might do arbitrary malicious behavior. Indeed, consider the $n$-partite resource theory induced by $\CC^n\xrightarrow{\otimes}\CC\xrightarrow{\hom(I,-)}\Set$. Let us first model a situation where the first $n-1$ participants are honest and the last participant is dishonest. In this case we can set $\A=\prod_{i=1}^n \A_i$ where each of $\A_1,\dots ,\A_{n-1}$ is the minimal attack model on $\CC$ and $\A_n$ is the maximal attack model. Then, an attack on $\f=(f_1,\dots f_n)\colon ((A_i)_{i=1}^n,r)\to ((B_i)^n_{i=1},s)$ can be represented by the first $n-1$ parties obeying the protocol and the $n$-th party doing an arbitrary computation~$a$, as depicted in the two pictures of \cref{fig:attack},
where $[n]:=\{1,\dots,n\}$, $(k,n]:=\{k+1,\dots n\}$, $\f|_{[k]}:=\bigotimes_{i=1}^k f_i$, and here $k=n-1$. The latter representation will be used when we do not need to emphasize pictorially the fact that the honest parties are each performing their own individual computations.

If instead of just one attacker, there are several \emph{independently} acting adversaries, we can take $\A=\prod_{i=1}^n \A_i$ where $\A_i$ is the minimal or maximal attack structure depending on whether the $i$th participant is honest or not. If the set of dishonest parties can collude and communicate arbitrarily during the process, we need the flexibility given in Definition~\ref{def:attack} and have the attack structure live in a different category than where our protocols live. For simplicity of notation, assume that the first~$k$ agents are honest but the remaining parties are malicious and might do arbitrary (joint) processes in $\CC$. In particular, the action done by the dishonest parties $k+1,\dots , n$ need not be describable as a product $\bigotimes_{i=k+1}^n (a_i)$ of individual actions. In that case we define $\A$ as follows: we first consider our resource theory as arising from $\CC^n\xrightarrow{\id^k\times\otimes}\CC^k\times \CC\xrightarrow{\otimes}\CC\xrightarrow{\hom(I,-)}\Set$, and define $\A$ on $\CC^k\times \CC$ as the product of the minimal attack model on $\CC^k$ and the maximal one on $\CC$. Concretely, this means that the first $k$ agents always obey the protocol, but the remaining agents can choose to perform arbitrary joint behaviors in $\CC$. Then a generic attack on a protocol $\f$ can be represented exactly as before in \cref{fig:attack}, except we no longer insist that $k=n-1$. Now a protocol $\f$ is $\A$-secure if for any $a$ with $\dom(a)=(A_i)_{i=k+1}^n$ there is a $b$ with $\dom(b)=(B_i)_{i=k+1}^n$ satisfying the equation of \cref{fig:security}.

If one is willing to draw more wire crossings, one can easily depict and define security against an arbitrary subset of the parties behaving maliciously, and henceforward this is the attack model we have in mind when we say that some $n$-partite protocol is secure against some subset of the parties. Moreover, for any subset $J$ of dishonest agents, one could consider more limited kinds of attacks: for instance, the agents might have limited computational power or limited abilities to perform joint computations---as long as the attack model satisfies the conditions of Definition~\ref{def:attack} one automatically gets a composable notion of secure protocols by Theorem~\ref{thm:composition} below.
\begin{theorem}\label{thm:composition} Given symmetric monoidal functors $F\colon\cat{D}\to\cat{C}$, $R\colon \cat{C}\to\Set$ with $F$ strong monoidal and $R$ lax monoidal, and an attack model $\A$ on $\cat{C}$, the class of $\A$-secure maps forms a wide sub-SMC of the resource theory $\res RF$ induced by $RF$.
\end{theorem}
\noindent So far we have discussed security only against a single, fixed subset of dishonest parties, while in multi-party computation it is common to consider security against any subset containing \eg at most $n/3$ or $n/2$ of the parties. However, as monoidal subcategories are closed under intersection, we immediately obtain composability against multiple attack models.
\begin{corollary}\label{cor:simultaneoussafety}
Given a non-empty family of functors $(\cat{D}\xrightarrow{F_i}\cat{C_i}\xrightarrow{R_i}\Set)_{i\in I}$ with $R_iF_i=R_jF_j=:R$ for all $i,j\in I$ and attack models $\A_i$ on $\CC_i$ for each $i$, the class of maps in $\res R$ that is secure against each $\A_i$ is a sub-SMC of $\res R$.
\end{corollary}
Using Corollary~\ref{cor:simultaneoussafety} one readily obtains composability of protocols that are simultaneously secure against different attack models $\A_i$. Thus one could, in principle, consider composable cryptography in an $n$-party setting where some subsets are honest-but-curious, some might be outright malicious but have limited computational power, and some subsets might be outright malicious but not willing or able to coordinate with each other, without reproving any composition theorems.

While the security definition of $f$ quantifies over $\A(f)$, which may be infinite, under suitable conditions it is sufficient to check security only on a subset of $\A(f)$, so that whether $f$ is $\A$-secure often reduces to finitely many equations.
\begin{definition}\label{def:initialattack}
Given $f\colon A\to B$, a subset $X$ of $\A(f)$ is said to be \emph{initial} if any $f'\in\A(f)$ with $\dom(f')=A$ can be factorized as $b\circ a$ with $a\in X$ and $b\in\A(\id[B])$.
\end{definition}
\begin{theorem}\label{thm:initialattacks}
 Let $f\colon (A,r)\to (B,s)$ define a morphism in the resource theory induced by $F\colon\cat{D}\to\cat{C}$ and $R\colon \cat{C}\to\cat{Set}$ and let $\A$ be an attack model on $\CC$. If $X\subset\A(F(f))$ is initial, then $f$ is $\A$-secure if, and only if the security condition holds against attacks in $X$, \ie if for any $f'\in X$ with $\dom(f')=F(A)$ there is $b\in\A(\id[F(B)])$ such that $R(f')r=R(b)s$.
\end{theorem}
Let us return to the example of $\CC^n\to\CC$ with the first $k$ agents being honest and the final $n-k$ dishonest and collaborating. Then we can take a singleton as our initial subset of attacks on $\f$, and this is given by $\f|_{[k]}\otimes (\bigotimes_{i=k+1}^n\id)$. Intuitively, this represents a situation where the dishonest parties $k+1,\dots ,n$ merely stand by and forward messages between the environment and the functionality, so that initiality can be seen as explaining ``completeness of the dummy adversary''~\cite[Claim 11]{Can01} in UC-security. In this case the security condition can be equivalently phrased by saying that there exists $b\in\A([\id[b]])$ satisfying the equation of \cref{fig:initial_attack},
which reproduces the pictures of~\cite{MT13}. Similarly, for classical honest-but-curious adversaries one usually only considers the initial such adversary, who follows the protocol otherwise except that they keep track of the protocol transcript.
\begin{theorem}\label{thm:perfectlifting}
In the resource theory of $n$-partite states, if $(f_1,\dots f_n)$ is secure against some subset $J$ of $[n]$ and $F$ is a strong monoidal, then $(Ff_1,\dots, Ff_n)$ is secure against~$J$ as well.
\end{theorem}
\noindent For instance, if the inclusion of classical interactive computations into quantum ones is strong monoidal, \ie respects sequential and parallel composition (up to isomorphism), then unconditionally secure classical protocols are also secure in the quantum setting, as shown in the context of UC-security in~\cite[Theorem 15]{Unr10}. More generally, this result implies that the construction of the category of $n$-partite transformations secure against any fixed subset of $[n]$ is functorial in $\cat{C}$, and this is in fact also true for any family of subsets of $[n]$ by Corollary~\ref{cor:simultaneoussafety}.
\section{Computational security}\label{sec:extensions}
The discussion above has been focused on perfect security, so that the equations defining security hold exactly. This is often too high a standard for security to hope for, and consequently cryptographers routinely work with computational or approximate security. We model this in two ways. The first approach replaces equations with an equivalence relation abstracting from the idea that the end results are ``computationally indistinguishable'' rather than strictly equal. The latter approach amounts to working in terms of a (pseudo)metric quantifying how close we are to the ideal resource and is needed to model statements in finite-key cryptography~\cite{TLGR12}. The typical metric is given by ``distinguisher advantage for polynomial-time environments'', enabling one to use computational complexity theory. In a nutshell, this amounts to working with sequences of protocols and defining security by saying ``for any $\epsilon>0$, for sufficiently large $n$, for any attack on the $n$th protocol there is an attack on the target resource such that the end results are within $\epsilon$''. The first approach is mathematically straightforward and we discuss it next, while the second approach is relegated to an extended version.

Replacing strict equations with equivalence relations is easy to describe on an abstract level as an instance of the theory so far: one just assumes that $\CC$ has a monoidal congruence $\approx$ and then works with the resource theory induced by $\CC^n\to\CC/{\approx}\xrightarrow{\hom(I,-)}\Set$ with similar attack models as above. More explicitly, as long as each hom-set of $\CC$ is equipped with an equivalence relation $\approx$ that respects $\otimes$ and $\circ$ in that $f\approx f'$ and $g\approx g'$ imply $gf\approx g'f'$ (whenever defined) and $g\otimes f\approx g'\otimes f'$, then working with $\CC^n\to\CC/{\approx}\xrightarrow{\hom(I,-)}\Set$ results in security conditions that replace $=$ in $\CC$ with $\approx$ throughout. If $\CC$ describes (interactive) computational processes and $\approx$ represents computational indistinguishability (inability for any ``efficient'' process to distinguish between the two), one might need to replace $\CC$ (and consequently functionalities, protocols and attacks on them) with the subcategory of $\CC$ of efficient processes so that $\approx$ indeed results in a congruence.
\section{Applications}\label{sec:applications}
We will now explore how the one-time pad (OTP) fits into our framework, paralleling the discussion of OTP in~\cite{Mau11}. We will start from the category $\cat{FinStoch}$ of finite sets and stochastic maps between them, with $\otimes$ given by cartesian product of sets. This is sufficient for OTP, even if more complicated and interactive cryptographic protocols will need a different starting category. However, the actual category $\CC$ we work in is built from $\cat{FinStoch}$, essentially by a tripartite variant of the ``resource theory of universally-combinable processes'' of~\cite[Section 3.4]{CFS16}. We will defer the detailed construction of $\CC$ to an extended version and work in it more heuristically, allowing us to focus on the OTP.

Roughly speaking, a ``basic object'' of $\CC$ consists of finite sets $A_i$,$B_i,E_i$ for $i=1,2$, and of a map $f\colon A_1\otimes B_1\otimes E_1\to A_2\otimes B_2\otimes E_2$ in $\cat{FinStoch}$, depicted in \cref{fig:sharedbox}.
\begin{figure}
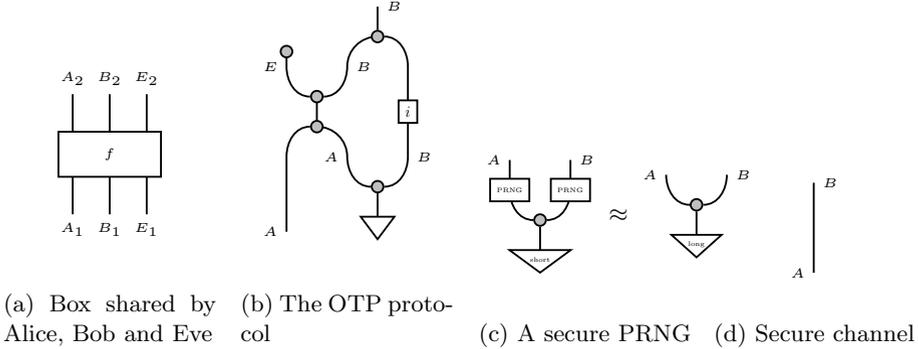
\vspace{-7.5mm}
    \centering
\begin{subfigure}[b]{0.23\textwidth}
\[
  \begin{pic}
    \setlength\minimummorphismwidth{10mm}
    \node[morphism] (f) at (0,0) {$f$};
    \draw ([xshift=-2.5pt]f.south west) to +(0,-.5) node[below] {$A_1$};
    \draw ([xshift=-2.5pt]f.north west) to +(0,.5) node[above] {$A_2$};
    \draw (f.south) to +(0,-.5) node[below] {$B_1$};
    \draw (f.north) to +(0,.5) node[above] {$B_2$};
    \draw ([xshift=2.5pt]f.south east) to +(0,-.5) node[below] {$E_1$};
    \draw ([xshift=2.5pt]f.north east) to +(0,.5) node[above] {$E_2$};
  \end{pic}\]
  \caption{Box shared by Alice, Bob and Eve}
  \label{fig:sharedbox}
  \end{subfigure}
  ~
    \begin{subfigure}[b]{0.23\textwidth}
\[
  \begin{pic}[scale=.4]
    \node[dot] (d) {};
    \draw (d) to +(0,-1) node[state,scale=0.5] {};
    \draw (d) to[out=0,in=-90] ++(1,1) node[right] {$B$} to ++(0,1.5) node[morphism,scale=.5] (i) {\large $i$};
    \draw (d) to[out=180,in=-90] ++(-1,1) node[left] {$A$} to[out=90,in=0] ++(-1,1) node[dot] (e) {};
    \draw (e) to[out=180,in=90] ++(-1,-1) to ++(0,-2.5) node[left] {$A$};
    \draw (e) to ++(0,1) node[dot] (f) {};
    \draw (f) to[out=180,in=-90] ++(-1,1) node[left] {$E$} to ++(0,.5) node[dot] {};
    \draw (f) to[out=0,in=-90] ++(1,1) node[right] {$B$} to[out=90,in=180] ++(1,1) node[dot] (g) {};
    \draw (g) to +(0,1) node[right] {$B$};
    \draw (g) to[out=0,in=90] ++(1,-1) to (i);
  \end{pic}
  \]
        \caption{The OTP protocol}
        \label{fig:OTP}
    \end{subfigure}
    ~
    \begin{subfigure}[b]{0.23\textwidth}
\[  \begin{pic}[scale=.4,yscale=-1]
    \node[dot] (d) {};
    \draw (d) to +(0,1) node[state,scale=0.5] {short};
    \draw (d) to[out=0,in=90] ++(1,-1) node[morphism,scale=0.5] {PRNG} to ++(0,-1) node[right] {$B$};
    \draw (d) to[out=180,in=90] ++(-1,-1)  node[morphism,scale=0.5] {PRNG} to ++(0,-1) node[left] {$A$};
  \end{pic}\approx
  \begin{pic}[scale=.4,yscale=-1]
    \node[dot] (d) {};
    \draw (d) to +(0,1) node[state,scale=0.5] {long};
    \draw (d) to[out=0,in=90] ++(1,-1) node[right] {$B$};
    \draw (d) to[out=180,in=90] ++(-1,-1) node[left] {$A$};
  \end{pic}
  \]
        \caption{A secure PRNG}
        \label{fig:PRNG}
    \end{subfigure}
   ~\begin{subfigure}[b]{0.23\textwidth}
\[
\begin{pic}[scale=.4] \draw (0,0) node[left] {$A$} to (0,3) node[right] {$B$}; \end{pic}  \]
        \caption{Secure channel}
        \label{fig:secure_channel}
   \end{subfigure}
\caption{Some resources and protocols\vspace{-1em}}
\end{figure}
The intuition is that $\tuple{(A_i,B_i,E_i)_{i\in\{1,2\}},f}$ represents a box shared by Alice, Bob and Eve, with Alice's inputs and outputs ranging over $A_1$ and $A_2$ respectively, and similarly for Bob and Eve. We will often label the ports just by the party who controls it, and omit labeling trivial ports. For example, if \cref{fig:copy_map}
\begin{figure}
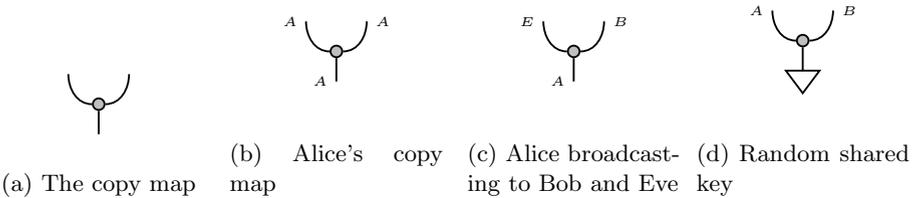
\vspace{-7.5mm}
    \centering
\begin{subfigure}[b]{0.23\textwidth}
\[
  \begin{pic}[scale=.4,yscale=-1]
    \node[dot] (d) {};
    \draw (d) to +(0,1);
    \draw (d) to[out=0,in=90] ++(1,-1);
    \draw (d) to[out=180,in=90] +(-1,-1);
  \end{pic}
  \]
  \caption{The copy map}
  \label{fig:copy_map}
  \end{subfigure}
  ~
    \begin{subfigure}[b]{0.23\textwidth}
\[
  \begin{pic}[scale=.4,yscale=-1]
    \node[dot] (d) {};
    \draw (d) to +(0,1) node [left] {$A$};
    \draw (d) to[out=0,in=90] ++(1,-1) node[right] {$A$};
    \draw (d) to[out=180,in=90] +(-1,-1) node[left] {$A$};
  \end{pic}
  \]
        \caption{Alice's copy map}
        \label{fig:private_copy}
    \end{subfigure}
    ~
    \begin{subfigure}[b]{0.23\textwidth}
\[
  \begin{pic}[scale=.4,yscale=-1]
    \node[dot] (d) {};
    \draw (d) to +(0,1) node [left] {$A$};
    \draw (d) to[out=0,in=90] ++(1,-1) node[right] {$B$};
    \draw (d) to[out=180,in=90] ++(-1,-1) node[left] {$E$};
  \end{pic}
  \]
        \caption{Alice broadcasting to Bob and Eve}
        \label{fig:insecure_channel}
    \end{subfigure}
   ~\begin{subfigure}[b]{0.23\textwidth}
\[
  \begin{pic}[scale=.4,yscale=-1]
    \node[dot] (d) {};
    \draw (d) to +(0,1) node[state,scale=0.5] {};
    \draw (d) to[out=0,in=90] +(1,-1)  node[right] {$B$};
    \draw (d) to[out=180,in=90] +(-1,-1) node[left] {$A$};
  \end{pic}
  \]
        \caption{Random shared key}
        \label{fig:shared_key}
   \end{subfigure}
   \caption{Variants of the copy map\vspace{-1em}}
\end{figure}
depicts the copy map $X\to X\otimes X$ for some set $X$ in $\cat{FinStoch}$, then \cref{fig:private_copy} denotes an object of $\CC$ representing Alice copying data privately, whereas \cref{fig:insecure_channel} denotes an object $\CC$ that sends Alice's input unchanged to Bob and to Eve---which we view as an insecure (but authenticated) channel from Alice to Bob.

A general object of $\CC$ then consists of a list of such basic objects, representing a list of such resources shared between Alice, Bob and Eve. A morphism of $\CC$ is roughly speaking a way of using the starting resources and local computation by the three parties to produce the target resources: a more formal description will be given in an extended version. In our attack model Alice and Bob are honest but Eve is dishonest, so she might do arbitrary local computation instead of whatever our protocols might prescribe.

In the version of the OTP we discuss, our starting resources consist of an insecure but authenticated channel\footnote{If the insecure channel allows Eve to tamper with the message, the analysis changes.} from Alice to Bob as in \cref{fig:insecure_channel} and (\ie $\otimes$) of a random key over the same message space, shared by Alice and Bob (\cref{fig:shared_key}). The goal is to build a secure channel from Alice to Bob (\cref{fig:secure_channel}) from these.

The local ingredients of OTP and the axioms they obey are depicted in \cref{fig:hopf} and correspond to a Hopf algebra with an integral in a SMC. Any finite group gives rise to such a structure in $\cat{FinStoch}$, with the integral given by the uniform distribution. Concretely, this means that Alice and Bob must agree on a group structure on the message space, and the fact that this multiplication forms a group and that the key is random can be captured by the equations of \cref{fig:hopf}.
\begin{figure}
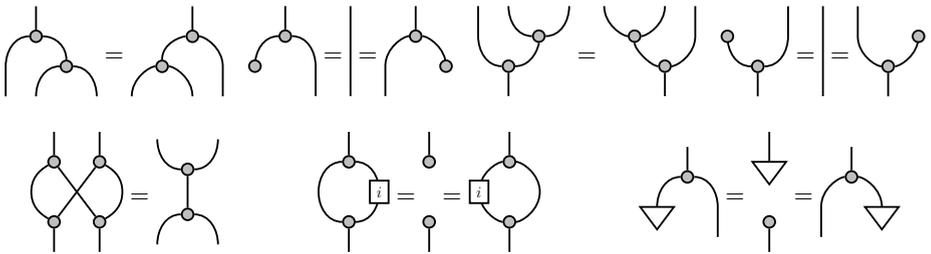
\vspace{-7.5mm}
    \centering
      \[
    \begin{pic}[scale=.4]
      \node[dot] (t) at (0,1) {};
      \node[dot] (b) at (1,0) {};
      \draw (t) to +(0,1);
      \draw (t) to[out=0,in=90] (b);
      \draw (t) to[out=180,in=90] (-1,0) to (-1,-1);
      \draw (b) to[out=180,in=90] (0,-1);
      \draw (b) to[out=0,in=90] (2,-1);
    \end{pic}
    =
    \begin{pic}[yscale=.4,xscale=-.4]
      \node[dot] (t) at (0,1) {};
      \node[dot] (b) at (1,0) {};
      \draw (t) to +(0,1);
      \draw (t) to[out=0,in=90] (b);
      \draw (t) to[out=180,in=90] (-1,0) to (-1,-1);
      \draw (b) to[out=180,in=90] (0,-1);
      \draw (b) to[out=0,in=90] (2,-1);
    \end{pic}
  \quad
  \begin{pic}[scale=.4]
    \node[dot] (d) {};
    \draw (d) to +(0,1);
    \draw (d) to[out=0,in=90] ++(1,-1) to ++(0,-1);
    \draw (d) to[out=180,in=90] ++(-1,-1) node[dot] {};
  \end{pic}
  =
  \begin{pic}[scale=.4]
    \draw (0,0) to (0,3);
  \end{pic}
  =
  \begin{pic}[yscale=.4,xscale=-.4]
    \node[dot] (d) {};
    \draw (d) to ++(0,1);
    \draw (d) to[out=0,in=90] ++(1,-1) to ++(0,-1);
    \draw (d) to[out=180,in=90] ++(-1,-1) node[dot] {};
  \end{pic}
  \quad
    \begin{pic}[xscale=.4,yscale=-.4]
      \node[dot] (t) at (0,1) {};
      \node[dot] (b) at (1,0) {};
      \draw (t) to ++(0,1);
      \draw (t) to[out=0,in=90] (b);
      \draw (t) to[out=180,in=90] (-1,0) to (-1,-1);
      \draw (b) to[out=180,in=90] (0,-1);
      \draw (b) to[out=0,in=90] (2,-1);
    \end{pic}
    =
    \begin{pic}[yscale=-.4,xscale=-.4]
      \node[dot] (t) at (0,1) {};
      \node[dot] (b) at (1,0) {};
      \draw (t) to ++(0,1);
      \draw (t) to[out=0,in=90] (b);
      \draw (t) to[out=180,in=90] (-1,0) to (-1,-1);
      \draw (b) to[out=180,in=90] (0,-1);
      \draw (b) to[out=0,in=90] (2,-1);
    \end{pic}
  \quad
  \begin{pic}[scale=.4,yscale=-1]
    \node[dot] (d) {};
    \draw (d) to ++(0,1);
    \draw (d) to[out=0,in=90] ++(1,-1) to ++(0,-1);
    \draw (d) to[out=180,in=90] ++(-1,-1) node[dot] {};
  \end{pic}
  =
  \begin{pic}[scale=.4]
    \draw (0,0) to (0,3);
  \end{pic}
  =
  \begin{pic}[yscale=-.4,xscale=-.4]
    \node[dot] (d) {};
    \draw (d) to ++(0,1);
    \draw (d) to[out=0,in=90] ++(1,-1) to ++(0,-1);
    \draw (d) to[out=180,in=90] ++(-1,-1) node[dot] {};
  \end{pic}
  \]
  \[
  \begin{pic}[scale=.4]
    \node[dot] (d) at (-0.75,1) {};
    \node[dot] (e) at (-0.75,-1) {};
    \draw (d) to ++(0,1);
    \draw (d) to[out=210,in=90] ++(-.75,-1) to[out=-90,in=150] (e);
    \draw (e) to ++(0,-1);
    \node[dot] (d1) at (0.75,1) {};
    \node[dot] (e1) at (0.75,-1) {};
    \draw (d1) to[out=-30,in=90] ++(.75,-1) to[out=-90,in=30] (e1);
    \draw (e1) to ++(0,-1);
    \draw (d1) to ++(0,1);
    \draw (d) to (e1);
    \draw (e) to (d1);
  \end{pic}
  =
  \begin{pic}[scale=.4]
    \node[dot] (r) at (0,.75) {};
    \node[dot] (d) at (0,-.75) {};
    \draw (d) to (r);
    \draw (d) to[out=0,in=90] ++(1,-1);
    \draw (d) to[out=180,in=90] ++(-1,-1);
    \draw (r) to[out=0,in=-90] ++(1,1);
    \draw (r) to[out=180,in=-90] ++(-1,1);
  \end{pic} \qquad\qquad
  \begin{pic}[scale=.4]
    \node[dot] (d) at (0,1) {};
    \node[dot] (e) at (0,-1) {};
    \draw (d) to ++(0,1);
    \draw (d) to[out=0,in=90] ++(1,-1)  node[morphism,scale=.5] {\large $i$} to[out=-90,in=0] (e);
    \draw (d) to[out=180,in=90] ++(-1,-1) to[out=-90,in=180] (e);
    \draw (e) to ++(0,-1);
  \end{pic}
  =
  \begin{pic}[scale=.4]
    \node[dot] (r) at (0,1) {};
    \node[dot] (d) at (0,-1) {};
    \draw (d) to ++(0,-1);
    \draw (r) to ++(0,1);
  \end{pic}
  =
  \begin{pic}[scale=.4,xscale=-1]
    \node[dot] (d) at (0,1) {};
    \node[dot] (e) at (0,-1) {};
    \draw (d) to +(0,1);
    \draw (d) to[out=0,in=90] ++(1,-1)  node[morphism,scale=.5] {\large $i$} to[out=-90,in=0] (e);
    \draw (d) to[out=180,in=90] ++(-1,-1) to[out=-90,in=180] (e);
    \draw (e) to ++(0,-1);
  \end{pic}
 \qquad\qquad
    \begin{pic}[scale=.4]
    \node[dot] (d) {};
    \draw (d) to ++(0,1);
    \draw (d) to[out=0,in=90] ++(1,-1) to ++(0,-1);
    \draw (d) to[out=180,in=90] ++(-1,-1) node[state,scale=0.5] {};
  \end{pic}
  =
  \begin{pic}[scale=.4]
    \node[state,scale=0.5] (r) at (0,1) {};
    \node[dot] (d) at (0,-1) {};
    \draw (d) to ++(0,-1);
    \draw (r) to ++(0,1);
  \end{pic}
  =
  \begin{pic}[yscale=.4,xscale=-.4]
    \node[dot] (d) {};
    \draw (d) to ++(0,1);
    \draw (d) to[out=0,in=90] ++(1,-1) to ++(0,-1);
    \draw (d) to[out=180,in=90] ++(-1,-1) node[state,scale=0.5] {};
  \end{pic}\]
  \caption{Local ingredients of OTP and the axioms they obey \vspace{-1em}}
  \label{fig:hopf}
\end{figure}

The OTP protocol is then depicted in \cref{fig:OTP}, \ie Alice adds the key to her message, broadcasts it to Eve and Bob. Eve deletes her part and Bob adds the inverse of the key to the ciphertext to 
recover the message.

To show that the protocol is secure, note that Eve has an initial attack given by just reading the ciphertext. The pictorial security proof is depicted in~\cref{fig:securityofOTP}.
\begin{figure}
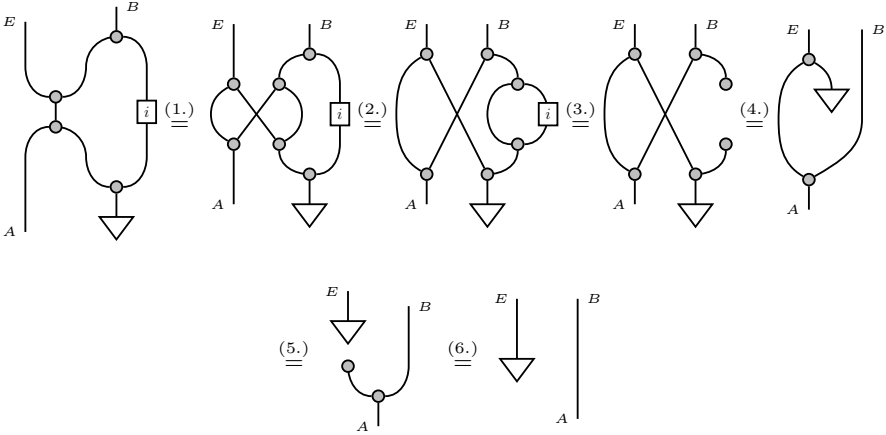

    \centering
    \[
  \begin{pic}[scale=.4]
    \node[dot] (d) {};
    \draw (d) to +(0,-1) node[state,scale=0.5] {};
    \draw (d) to[out=0,in=-90] ++(1,1) to ++(0,1.5) node[morphism,scale=.5] (i) {\large $i$};
    \draw (d) to[out=180,in=-90] ++(-1,1) to[out=90,in=0] ++(-1,1) node[dot] (e) {};
    \draw (e) to[out=180,in=90] ++(-1,-1) to  ++(0,-2.5) node[left] {$A$};
    \draw (e) to ++(0,1) node[dot] (f) {};
    \draw (f) to[out=180,in=-90] ++(-1,1)  to ++(0,1.5) node[left] {$E$};
    \draw (f) to[out=0,in=-90] ++(1,1) to[out=90,in=180] ++(1,1) node[dot] (g) {};
    \draw (g) to +(0,1) node[right] {$B$};
    \draw (g) to[out=0,in=90] ++(1,-1) to (i);
  \end{pic}
   \overset{(1.)}{=}
    \begin{pic}[scale=.4]
    \node[dot] (d) at (-0.75,1) {};
    \node[dot] (e) at (-0.75,-1) {};
    \draw (d) to ++(0,1) to ++(0,1) node[left] {$E$} ;
    \draw (d) to[out=210,in=90] ++(-.75,-1) to[out=-90,in=150] (e);
    \draw (e) to ++(0,-1)  to ++(0,-1) node[left] {$A$};
    \node[dot] (d1) at (0.75,1) {};
    \node[dot] (e1) at (0.75,-1) {};
    \draw (d) to (e1);
    \draw (e) to (d1);
    \draw (d1) to[out=-30,in=90] ++(.75,-1) to[out=-90,in=30] (e1);
    \draw (e1) to[out=-90,in=180] ++(1,-1) node[dot] (f) {};
    \draw (d1) to[out=90,in=180] ++(1,1) node[dot] (g) {};
    \draw (f) to ++(0,-1) node[state,scale=.5] {};
    \draw (f) to[out=0,in=-90] ++(1,1) to ++(0,1) node[morphism,scale=.5] (i) {\large $i$};
    \draw (g) to[out=0,in=90] ++(1,-1) to (i);
    \draw (g) to ++(0,1) node[right] {$B$};
  \end{pic}
   \overset{(2.)}{=}
    \begin{pic}[scale=.4]
    \node[dot] (d) at (1,1) {};
    \node[dot] (e) at (1,-1) {};
    \node[dot] (a) at (-2,2) {};
    \node[dot] (b) at (-2,-2) {};
    \draw (d) to[out=90,in=0] ++(-1,1) node[dot] (f) {};
    \draw (d) to[out=0,in=90] ++(1,-1)  node[morphism,scale=.5] {\large $i$} to[out=-90,in=0] (e);
    \draw (d) to[out=180,in=90] ++(-1,-1) to[out=-90,in=180] (e);
    \draw (e) to[out=-90,in=0] ++(-1,-1) node[dot] (g) {};
    \draw (f) to +(0,1) node[right] {$B$};
    \draw (f) to (b);
    \draw (g) to +(0,-1) node[state,scale=.5] {};
    \draw (g) to (a);
    \draw (a) to +(0,1) node[left] {$E$};
    \draw (b) to +(0,-1) node[left] {$A$};
    \draw (a) to[in=90,out=210] ++(-1,-2) to[in=150,out=-90] (b);
  \end{pic}
  \overset{(3.)}{=}
      \begin{pic}[scale=.4]
    \node[dot] (d) at (1,1) {};
    \node[dot] (e) at (1,-1) {};
    \node[dot] (a) at (-2,2) {};
    \node[dot] (b) at (-2,-2) {};
    \draw (d) to[out=90,in=0] ++(-1,1) node[dot] (f) {};
    \draw (e) to[out=-90,in=0] ++(-1,-1) node[dot] (g) {};
    \draw (f) to +(0,1) node[right] {$B$};
    \draw (f) to (b);
    \draw (g) to +(0,-1) node[state,scale=.5] {};
    \draw (g) to (a);
    \draw (a) to +(0,1) node[left] {$E$};
    \draw (b) to +(0,-1) node[left] {$A$};
    \draw (a) to[in=90,out=210] ++(-1,-2) to[in=150,out=-90] (b);
  \end{pic}
  \overset{(4.)}{=}
    \begin{pic}[scale=.4]
    \node[dot] (d) at (0,2) {};
    \node[dot] (e) at (0,-2) {};
    \draw (d) to +(0,1) node[left] {$E$};
    \draw (d) to[out=-30,in=90] ++(.75,-1)  node[state,scale=.5] {};
    \draw (d) to[in=90,out=210] ++(-1,-2) to[in=150,out=-90] (e);
    \draw (e) to +(0,-1) node[left] {$A$};
    \draw (e) to[out=30,in=-90] ++(1.75,2) to ++(0,3) node[right] {$B$};
  \end{pic}
\]
\[
 \overset{(5.)}{=}
   \begin{pic}[scale=.4]
  \node[state,scale=.5] (d) at (0,1.5) {};
    \draw (d) to +(0,1) node[left] {$E$};
    \node[dot] (a) at (1,-1) {};
    \draw (a) to[out=180,in=-90] ++(-1,1) node[dot] {};
    \draw (a) to ++(0,-1) node[left] {$A$};
    \draw (a) to[out=0,in=-90] ++(1,1) to ++(0,2) node[right] {$B$};
  \end{pic}
 \overset{(6.)}{=}
 \begin{pic}[scale=.4]
    \draw (0,0) node[left] {$A$} to (0,4) node[right] {$B$};
    \node[state,scale=.5] (d) at (-2,2) {};
    \draw (d) to +(0,2) node[left] {$E$};
  \end{pic}
  \]
  \caption{Security proof of OTP}
  \label{fig:securityofOTP}
\end{figure}
The first equation is the interaction between multiplication and copying, the second uses (co)associativity, the third one properties of inverses, the fourth and last one use unitality, and the fifth one follows from the key being random. Taken together, these show that Eve's initial attack is equal to her just producing a random message herself with Alice and Bob sharing the target resource. The correctness of the protocol can be proven similarly. Thus OTP gives a map $\text{shared key}\otimes\text{authenticated channel}\to\text{secure channel}$ that is secure against Eve.

We now use this example to illustrate the use of the composition theorems. A major drawback of OTP, despite its perfect security, is the fact that one needs a key that is as long as the message. In practice, Alice and Bob might only share a short key and wish to promote it a long key. If they agree on a pseudo-random number generator (PRNG) with their key as the seed, they can map the short key to a longer key. If the PRNG is computationally secure, then the end-result is (computationally) indistinguishable from a long key, depicted in \cref{fig:PRNG},
where $\approx$ stands for computational indistinguishability. We envision the computational security of the chosen PRNG to be proven ``the usual way'' and not graphically---after all, we believe that our framework is there to supplement ordinary cryptographic reasoning and not to replace it. The PRNG then results in a (computationally) secure way of promoting a short shared key into a long shared key, and then the composition theorems guarantee that these protocols can be composed, resulting in the security of the stream cipher. 

Composable security is a stronger constraint than stand-alone security, and indeed many cryptographic functionalities are known to be impossible to achieve ``in the plain model'', \ie without set-up assumptions. A case in point is bit commitment, which was shown to be impossible in the UC-framework in~\cite{CF01}. This result was later generalized in~\cite{PR08} to show that any two-party functionality that can be realized in the plain UC-framework is ``splittable''. While the authors of~\cite{PR08} remark that their result applies more generally than just to the UC-framework, this wasn't made precise until~\cite{MR11}\footnote{Except that in their framework the 2-party case seems to require security constraints also when both parties cheat.}. We present a categorical proof of this result in our framework, which promotes the pictures ``illustrating the proof'' in~\cite{PR08} into a full proof---the main difference is that in~\cite{PR08} the pictures explicitly keep track of an environment trying to distinguish between different functionalities, whereas we prove our result in the case of perfect security and then deduce the asymptotic claim.

We now assume that $\CC$, our ambient category of interactive computations is compact closed\footnote{We do not view this as overtly restrictive, as many theoretical models of concurrent interactive (probabilistic/quantum) computation are compact closed~\cite{winskel:game,clairambault:gamesforquantum,clairambaultetal:gamesforquantum2}.}. As we are in the 2-party setting, we take our free computations to be given by $\CC^2$, and we consider two attack models: one where Alice cheats and Bob is honest, and one where Bob cheats and Alice is honest. We think of $\tinycup$  as representing a two-way communication channel, but this interpretation is not needed for the formal result.
\begin{theorem}\label{thm:bipartite} For Alice and Bob (one of whom might cheat),
if a bipartite functionality $r$ can be securely realized from a communication channel between them, \ie from $\tinycup$, then there is a $g$ such that
\begin{equation}\label{eq:splittable}\tag{$*$}
\begin{pic}
    \node[state] (x) at (0,0) {$r$};
    \draw (x.A) to +(0,.4) node[left] {$A$};
    \draw (x.B) to +(0,.4) node[right] {$B$};
  \end{pic}=\begin{pic}\setlength\morphismheight{3mm}
  \setlength\minimummorphismwidth{3mm}
    \node[state] (x) at (0,0) {$r$};
     \node[state] (y) at (1,0) {$r$};
    \node[morphism] (f) at (0.5,0.4) {$g$};
    \draw (x.A) node[left] {$$} to +(0,.6);
    \draw (x.B) to node[right] {} ([xshift=-4.25pt] f.south west);
    \draw (y.A) to node[right] {} ([xshift=4.25pt] f.south east);
    \draw (y.B) to  node[right] {$$} +(0,.6);
  \end{pic}.
  \end{equation}
\end{theorem}
\begin{proof}
If a protocol $(f_A, f_B)$ achieves this, security constraints give us $s_A,s_B$ 
    \[\text{such that }\begin{pic}
    \node[morphism,scale=.5,font=\normalsize] (f) at (0,0) {$f_A$};
    \draw (f.north) to +(0,.3);
    \draw (f.south)  to[out=-90,in=-180] ++(0.25,-.25) to[out=0,in=-90] ++(.25,.25) to ++(0,.6);
    \end{pic}=\begin{pic}
     \node[morphism,scale=.5,font=\normalsize] (g) at (.225,.35) {$s_B$};
    \node[state] (x) at (0,0) {$r$};
    \draw (x.A) to +(0,.75) node[right] {};
    \draw (g.south) to node[right] {$$} (x.B);
    \draw (g.north) to +(0,.25) node[right] {};
  \end{pic}
 \quad \text{and} \quad\
  \begin{pic}[xscale=-1]
    \node[morphism,scale=.5,font=\normalsize] (f) at (0,0) {$f_B$};
    \draw (f.north) to +(0,.3);
    \draw (f.south)  to[out=-90,in=-180] ++(0.25,-.25) to[out=0,in=-90] ++(.25,.25) to ++(0,.6);
    \end{pic}=\begin{pic}[xscale=-1]
     \node[morphism,scale=.5,font=\normalsize] (g) at (.225,.35) {$s_A$};
    \node[state] (x) at (0,0) {$r$};
    \draw (x.B) to +(0,.75) node[right] {};
    \draw (g.south) to node[right] {$$} (x.A);
    \draw (g.north) to +(0,.25) node[right] {};
  \end{pic}  \]

   \[ \text{so that} \quad\begin{pic}
    \node[state] (x) at (0,0) {$r$};
    \draw (x.A) to +(0,.3) node[right] {};
    \draw (x.B) to +(0,.3) node[right] {};
  \end{pic}=\begin{pic}
    \node[morphism,scale=.5,font=\normalsize] (f) at (0,0) {$f_A$};
    \node[morphism,scale=.5,font=\normalsize] (g) at (.5,0) {$f_B$};
    \draw (f.north) to +(0,.3);
    \draw (g.north) to +(0,.3);
    \draw (f.south)  to[out=-90,in=180] ++(0.25,-.25) to[out=0,in=-90] (g.south);
    \end{pic}
  =\begin{pic}
    \node[morphism,scale=.5,font=\normalsize] (f) at (0,0) {$f_A$};
    \node[morphism,scale=.5,font=\normalsize] (g) at (1.5,0) {$f_B$};
    \draw (f.north) to +(0,.3);
    \draw (g.north) to +(0,.3);
    \draw (f.south)  to[out=-90,in=-180] ++(.25,-.25) to[out=0,in=-90] ++(.25,.25) to[out=90,in=180] ++(.25,.25) to [out=0,in=90] ++(.25,-.25) to[out=-90,in=180] ++(.25,-.25) to[out=0,in=-90] (g.south);
    \end{pic}%
  =
\begin{pic}
     \node[morphism,scale=.5,font=\normalsize] (g) at (.23,.35) {$s_B$};
    \node[state] (x) at (0,0) {$r$};
    \node[morphism,scale=.5,font=\normalsize] (f) at (0.78,.35) {$s_A$};
    \node[state] (y) at (1,0) {$r$};
    \draw (x.A) to +(0,.75) node[right] {};
    \draw (y.B) to +(0,.75) node[right] {};
    \draw (g.south) to node[right] {$$} (x.B);
    \draw (g.north) to[out=90,in=-180] ++(0.275,.25) to [out=0,in=90] (f.north);
    \draw (f.south) to node[right] {$$} (y.A);
  \end{pic}
  \]
\end{proof}
\begin{corollary}  Given a compact closed $\CC$  modeling  computation in which wires model communication channels, (composable) bit commitment and oblivious transfer are impossible in that model without setup, even asymptotically in terms of distinguisher advantage.
\end{corollary}
\begin{proof}
If $r$ represents bit commitment from Alice to Bob, it does not satisfy the equation required by Theorem~\ref{thm:bipartite} for any $g$, and the two sides of~\eqref{eq:splittable} can be distinguished efficiently with at least probability $1/2$. Indeed, take any $f$ and let us compare the two sides of~\eqref{eq:splittable}: if the distinguisher commits to a random bit $b$, then Bob gets a notification of this on the left hand-side, so that $f$ has to commit to a bit on the right side of~\eqref{eq:splittable} to avoid being distinguished from the left side. But this bit coincides with $b$ with probability at most $1/2$, so that the difference becomes apparent at the reveal stage. The case of OT is similar.
\end{proof}
We now discuss a similar result in the tripartite case, which rules out building a broadcasting channel from pairwise channels securely against any single party cheating. In~\cite{MMP+18} comparable pictures are used to illustrate the official, symbolically rather involved, proof, whereas in our framework the pictures are the proof. Another key difference is that~\cite{MMP+18} rules out broadcasting directly, whereas we show that any tripartite functionality realizable from pairwise channels satisfies some equations, and then use these equations to rule out broadcasting.

Formally, we are working with the resource theory given by $\CC^3\xrightarrow{\otimes}\CC\xrightarrow{\hom(I,-)}\Set$ where $\CC$ is an SMC, and reason about protocols that are secure against three kinds of attacks: one for each party behaving dishonestly while the rest obey the protocol. Note that we do not need to assume compact closure for this result, and the result goes through for any state on $A\otimes A$ shared between each pair of parties: we will denote such a state by $\tinycup$ by convention.
\begin{theorem}  If a tripartite functionality $r$ can be realized from each pair of parties sharing a state
 $\tinycup$, securely against any single party, then there are simulators $s_A,s_B,s_C$ such that
\[\begin{pic}
    \setlength\minimumstatewidth{15mm}
    \node[state] (r) at (0,0) {$r$};
    \node[morphism,scale=.5,font=\normalsize] (f) at (-.5,.35) {$s_A$};
    \draw ([xshift=-3.5pt]r.A) to (f.south);
    \draw ([xshift=3.5pt]r.B) to +(0,.75);
    \draw ([xshift=-.5pt]f.north west) to +(0,.25);
    \draw ([xshift=.5pt]f.north east) to +(0,.25);
    \draw (r.center) to +(0,.75);
  \end{pic}\quad=\quad
  \begin{pic}
    \setlength\minimumstatewidth{15mm}
    \node[state] (r) at (0,0) {$r$};
    \node[morphism,scale=.5,font=\normalsize] (h) at (0,.35) {$s_B$};
    \draw ([xshift=-3.5pt]r.A) to +(0,.75);
    \draw ([xshift=3.5pt]r.B) to +(0,.75);
    \draw ([xshift=-.5pt]h.north west) to +(0,.25);
    \draw ([xshift=.5pt]h.north east) to +(0,.25);
    \draw (r.center) to (h.south);
  \end{pic}\quad=\quad
  \begin{pic}
    \setlength\minimumstatewidth{15mm}
    \node[state] (r) at (0,0) {$r$};
    \node[morphism,scale=.5,font=\normalsize] (g) at (.5,.35) {$s_C$};
    \draw ([xshift=-3.5pt]r.A) to +(0,.75);
    \draw ([xshift=3.5pt]r.B) to (g.south);
     \draw ([xshift=-.5pt]g.north west) to +(0,.25);
     \draw ([xshift=.5pt]g.north east) to +(0,.25);
    \draw (r.center) to +(0,.75);
  \end{pic}\ .
  \]
\end{theorem}
\begin{proof}
Any tripartite protocol building on top of each pair of parties sharing $\tinycup$  can be drawn as in the left side of
\[
\begin{pic}
    \setlength\morphismheight{3mm}
    \node[morphism] (f) at (0,0) {$f_A$};
    \node[morphism] (g) at (1,0) {$f_B$};
    \node[morphism] (h) at (2,0) {$f_C$};
    \draw (f.north) to +(0,.3);
    \draw (g.north) to +(0,.3);
    \draw (h.north) to +(0,.3);
    \draw (f.south east)  to[out=-90,in=-90] (g.south west);
    \draw (g.south east)  to[out=-90,in=-90] (h.south west);
    \draw (f.south west)  to[out=-90,in=-90] (h.south east);
    \end{pic} \qquad\qquad
\begin{pic}
    \setlength\morphismheight{3mm}
    \coordinate (a) at (0.5,-.9);
    \node[morphism] (e) at (-1,0) {$f_A$};
    \node[morphism] (f) at (0,0) {$f_B$};
    \node[morphism] (g) at (1,0) {$f_B$};
    \node[morphism] (h) at (2,0) {$f_C$};
    \draw (e.north) to +(0,.3);
    \draw (f.north) to +(0,.3);
    \draw (g.north) to +(0,.3);
    \draw (h.north) to +(0,.3);
    \draw (e.south east)  to[out=-90,in=-90] (f.south west);
    \draw (f.south east)  to[out=-90,in=-90] (g.south west);
    \draw (g.south east)  to[out=-90,in=-90] (h.south west);
    \draw (e.south west)  to[out=-90,in=180] (a) to[out=0,in=-90] (h.south east);
    \end{pic}
   \]
Consider now the morphism in $\CC$ depicted on the right: it can be seen as the result of three different attacks on the protocol $(f_A,f_B,f_C)$ in $\CC^3$: one where Alice cheats and performs $f_A$ and $f_B$ (and the wire connecting them), one where Bob performs $f_B$ twice, and one where Charlie performs $f_B$ and $f_C$. The security of $(f_A,f_B,f_C)$ against each of these gives the required simulators.
\end{proof}
\begin{corollary}  Given a SMC $\CC$ modeling interactive computation, and a state $\tinycup$ on $A\otimes A$ modeling pairwise communication, it is impossible  to build broadcasting channels securely (even asymptotically in terms of distinguisher advantage) from pairwise channels.
\end{corollary}
\begin{proof} We show that a channel $r$ that enables Bob to broadcast an input bit to Alice and Charlie never satisfies the required equations for any $s_A,s_B,s_C$. Indeed,  assume otherwise and let the environment
plug ``broadcast $0$'' and ``broadcast $1$'' to the two wires in the middle. The leftmost picture then says that Charlie receives $1$, the rightmost picture implies that Alice gets $0$ and the middle picture that Alice and Bob get the same output (if anything at all)---a contradiction. Indeed, one cannot satisfy all of these simultaneously with high probability, which rules out an asymptotic transformation.
\end{proof}
\section{Outlook}\label{sec:outlook}
We have presented a categorical framework providing a general, flexible and mathematically robust way of reasoning about composability in cryptography. Besides contributing a further approach to composable cryptography and potentially helping with cross-talk and comparisons between existing approaches~\cite{CKL+19}, we believe that the current work opens the door for several further questions.

First, due to the generality of our approach we hope that one can, besides honest and malicious participants, reason about more refined kinds of adversaries composably. Indeed, we expect that Definition~\ref{def:attack}
is general enough to capture \eg honest-but-curious adversaries\footnote{Heuristically speaking this is the case: an honest-but-curious attack on $g\circ f$ should be factorizable as one on $g$ and one on $f$, and similarly an honest-but-curious attack on $g\otimes f$ should be factorisable into ones on $g$ and~$f$ that then forward their transcripts to an attack on $\id\otimes\id$.}. It would also be interesting to see if this captures even more general attacks, \eg situations where the sets of participants and dishonest parties can change during the protocol. This might require understanding our axiomatization of attack models more structurally and perhaps generalizing it.  Does this structure (or a variant thereof) already arise in category theory? While we define an attack model on a category, perhaps one could define an attack model on a (strong) monoidal functor~$F$, the current definition being recovered when $F=\id$.

Second, we expect that rephrasing cryptographic questions categorically would enable more cross-talk between cryptography and other fields already using category theory as an organizing principle. For instance, many existing approaches to composable cryptography develop their own models of concurrent, asynchronous, probabilistic and interactive computations. As categorical models of such computation exist in the context of game semantics~\cite{winskel:game,clairambault:gamesforquantum,clairambaultetal:gamesforquantum2}, one is left wondering whether the models of the semanticists' could be used to study and answer cryptographic questions, or conversely if the models developed by cryptographers contain valuable insights for programming language semantics.

Besides working inside concrete models---which ultimately blends into ``just doing composable cryptography''---one could study axiomatically how properties of a category relate to cryptographic properties in it. As a specific conjecture in this direction, one might hope to talk about honest-but-curious adversaries at an abstract level using environment structures~\cite{coecke:environment}, that axiomatize the idea of deleting a system. Similarly, having agents purify their actions is an important tool in quantum cryptography~\cite{LC97}---can categorical accounts of purification~\cite{chiribella:purification,cunningham:purity,coecke:environment} elucidate this?

Finally, we hope to get more mileage out of the tools brought in with the categorical viewpoint. For instance, can one prove further no-go results pictorially? More specifically, given the impossibility results for two and three parties, one wonders if the ``only topology matters'' approach of string diagrams can be used to derive general impossibility results for $n$ parties sharing pairwise channels. Similarly, while diagrammatic languages have been used to reason about positive cryptographic results in the stand-alone setting~\cite{kissinger2017picture,breiner:graphicaldicrypto,breiner:selftesting}, can one push such approaches further now that composable security definitions have a clear categorical meaning? Besides the graphical methods, thinking of cryptography as a resource theory suggests using resource-theoretic tools such as monotones. While monotones have already been applied in cryptography~\cite{WW08:monotones}, a full understanding of cryptographically relevant monotones is still lacking.

\bibliographystyle{splncs04}
\bibliography{bib/full,bib/quantum,bib/new,bib/new2}

\vfill

{\small\medskip\noindent{\bf Open Access} This chapter is licensed under the terms of the Creative Commons\break Attribution 4.0 International License (\url{http://creativecommons.org/licenses/by/4.0/}), which permits use, sharing, adaptation, distribution and reproduction in any medium or format, as long as you give appropriate credit to the original author(s) and the source, provide a link to the Creative Commons license and indicate if changes were made.}

{\small \spaceskip .28em plus .1em minus .1em The images or other
third party material in this chapter are included in the\break
chapter's Creative Commons license, unless indicated otherwise in a
credit line to the\break material.~If material is not included in
the chapter's Creative Commons license and\break your intended use
is not permitted by statutory regulation or exceeds the
permitted\break use, you will need to obtain permission directly
from the copyright holder.}

\medskip\noindent\includegraphics{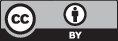}

\end{document}